\documentclass[lettersize,journal]{IEEEtran}
\usepackage{amsmath,amsfonts}
\usepackage{algorithmic}
\usepackage{algorithm}
\usepackage{array}
\usepackage[caption=false,font=normalsize,labelfont=sf,textfont=sf]{subfig}
\usepackage{textcomp}
\usepackage{stfloats}
\usepackage{url}
\usepackage{verbatim}
\usepackage{graphicx}
\usepackage{cite}
\usepackage{mathrsfs}

\usepackage{amssymb}
\usepackage{balance}
\usepackage{bm}
\usepackage{color}
\usepackage{threeparttable} 
\usepackage{diagbox}
\usepackage{makecell}
\usepackage{arydshln}

\newtheorem{thm}{Theorem}

\newtheorem{defin}{Definition}

\newtheorem{crly}{Corollary}
\newtheorem{eg}{Example}
\newtheorem{remark}{Remark}

\begin{document}

\title{SCMA Inspired Sparse Vector Coding: An Enhanced URLLC Transmission Scheme}

\author{Zhen-Ming Huang,~\IEEEmembership{Graduate Student Member,~IEEE}, Zilong Liu,~\IEEEmembership{Senior Member,~IEEE}, \\Leila Musavian,~\IEEEmembership{Senior Member,~IEEE}, and Chao-Yu~Chen,~\IEEEmembership{Senior Member,~IEEE}
\thanks{Zhen-Ming Huang is with the Institute of Computer and Communication Engineering,  National Cheng Kung University, Tainan 701, Taiwan, R.O.C. (e-mail: n98101012@gs.ncku.edu.tw).

Zilong Liu and Leila Musavian are with the School of Computer Science and Electronic Engineering, University of Essex, CO4 3SQ Colchester, U.K. (e-mail: zilong.liu@essex.ac.uk; leila.musavian@essex.ac.uk).
  			
Chao-Yu Chen is with the Department of Electrical Engineering and the Institute of Communications Engineering, National Cheng Kung University, Tainan 701, Taiwan, R.O.C. (e-mail: super@mail.ncku.edu.tw).
  	}
}


\IEEEpubid{0000--0000/00\$00.00~\copyright~2021 IEEE}

\maketitle

\begin{abstract}
Sparsity is inherently exploited in sparse code multiple access (SCMA) and sparse vector coding (SVC), yet the interaction between these two has not been explored before. It is intriguing to ask if one can be used to improve the other, and vice versa. In this work, we present a novel SCMA inspired SVC scheme, called SCMA-SVC, for enhanced ultra-reliable low-latency communications. Our key idea is to exploit the sparse pattern and multidimensional constellation nature of SCMA, with which one is able to further enlarge the minimum Euclidean distance (MED) of the corresponding SVC codebooks. Such an innovation allows us to harvest the multiuser coding gain and the constellation shaping gain which are pertinent to SCMA. Moreover, by applying random phase rotations to the sparse vectors, it is shown that the proposed SCMA-SVC achieves full diversity order over Rayleigh fading channels. Under maximum likelihood (ML) decoding, the proposed SCMA-SVC demonstrates  remarkable error rate performances over both Gaussian and Rayleigh fading channels. Additionally, we develop a low-complexity decoder that exploits the structural sparsity of SCMA-SVC while maintaining near-ML performance. Simulation results demonstrate that the proposed SCMA-SVC achieves significantly improved reliability over the existing SVC variants.
\end{abstract}

\begin{IEEEkeywords}
Sparse vector coding (SVC), sparse code multiple access (SCMA), short-packet transmission, ultra-reliable low-latency communications (URLLC).
\end{IEEEkeywords}

\section{Introduction}

\subsection{Background}

\IEEEPARstart{U}{ltra}-reliable low-latency communication (URLLC) is one of the major use cases in future communication systems, targeting mission-critical applications, such as autonomous vehicles, telemetry, remote surgery, and factory automation \cite{Ji_18_UELLC, Shirvanimoghaddam_19, Yue_24, Rowshan_24}. In such applications, the transmitted data are often in short packets, carrying certain high-priority command/control signals, life-critical information, or emergency rescue messages. In the fifth-generation (5G) networks, stringent performance requirements should be met by those URLLC applications, including a target block error rate (BLER) of $10^{-5}$ and an end-to-end latency constraint within $1$ ms. Moving towards 6G, IMT-2030 further extends URLLC to hyper-reliable and low-latency communication, targeting BLER as low as to $10^{-7}$ and latency of $0.1$-$1$ ms \cite{ITU_M2160_2023}. Conventional coding schemes for ultra-high reliability are feasible for sufficiently long block codes only. New design approaches are needed in view of the short-packet transmission nature and the stringent requirements to reliability and latency in URLLC. To this end, in the sequel, we first introduce sparse vector coding (SVC) and sparse code multiple access (SCMA), and then present SCMA inspired SVC (SCMA-SVC).

\IEEEpubidadjcol

\subsection{Related Works}
\subsubsection{SVC}

In 2018, SVC was proposed as a novel coding technique for short-packet URLLC transmission \cite{Ji_18_first}. In SVC, the basic idea is to map the information bits to certain nonzero elements of a sparse vector, with which one can generate the corresponding codeword by multiplying it with a predefined codebook matrix. The decoding of SVC can be regarded as a compressed sensing (CS) problem, whereby a codebook matrix with low mutual coherence is required to achieve reliable recovery. At the receiver, one may use multipath matching pursuit~\cite{Kwon_14} and orthogonal matching pursuit~\cite{Tropp_07} to recover the sparse vector from which the information bits are subsequently decoded. In recent years, SVC has been extended to various communication scenarios, including time division duplex systems \cite{Kim_TDD_20}, non-orthogonal multiple access \cite{Zhang_NOMA_21, Sabapathy_23}, multi-input multi-output systems \cite{Nam_MIMO_22,Zhang_MIMO_21}, orthogonal time frequency space modulation \cite{Ji_21}, integrated sensing and communication \cite{Zhang_ISAC_22}, high-mobility communications \cite{Zhang_Yanfeng_23}, and grant-free access \cite{Luo_24}. To further improve the reliability for URLLC transmissions, several SVC variants have been developed. In \cite{Kim_20} and \cite{Zhang_22}, enhanced SVC (ESVC) and constellation rotation-based SVC (CR-SVC) were proposed. Specifically, ESVC maps the information bits onto both the nonzero positions and their corresponding symbols of a sparse vector, whereas CR-SVC embeds a constellation rotation into the nonzero elements of the sparse vector. Subsequently, the index redefinition-based SVC (IR-SVC) was introduced in \cite{Zhang_23}, where the sparse vector mapping is redesigned through an index redefinition strategy. In general, these SVC variants are designed to facilitate more reliable and efficient sparse-vector recovery at the receiver, thereby improving transmission reliability and reducing decoding latency. Following these developments, numerous SVC variants have been proposed, with research efforts devoted to advanced SVC-based encoding schemes and corresponding decoding algorithms \cite{Zhang_23_China_COM,Yang_24,Mow_24_CWC,Zhang_24,Zhang_25,Zhang_25_SE_SVC, Zhang_25_block, Zhang_25_GLOBECOM, Zhang_25_One_Hot, Yang_26_RM}. In particular, a block SVC (BSVC) scheme was recently proposed by designing a block sparse mapping pattern to enhance the decoding performance \cite{Zhang_25_block}. In \cite{Mow_24_CWC}, constant weight codes (CWCs) were employed to construct the sparse vector set for SVC encoding, where the construction is based on the pairwise distances between sparse vectors. Another approach is to reduce the mutual coherence of the codebook matrix \cite{Fan_SVC_23, Yang_25, Huang_26_ICC}. A lower mutual coherence enables more accurate sparse vector recovery by CS decoders, thereby enhancing decoding performance. Specifically, the designs in \cite{Fan_SVC_23, Yang_25} were based on optimization frameworks that minimize the mutual coherence, whereas \cite{Huang_26_ICC} presented a structured construction based on periodic quasi-complementary pairs, leading to deterministically low mutual coherence.

\subsubsection{SCMA}

In parallel to the SVC research, sparsity is inherently incorporated into the SCMA codebook design through a predefined sparse matrix (similar to the sparse parity check matrices used in low-density parity check codes) \cite{Hosein_SCMA_first, Taherzadeh_SCMA_first}. Thanks to the codebook sparsity, message passing algorithm (MPA) can be exploited at the receiver to carry out efficient SCMA decoding. Secondly, SCMA is able to support higher spectrum efficiency (compared to legacy orthogonal multiple access schemes) due to overloaded multiuser transmission. In SCMA, multi-dimensional constellation design plays an important role for improved error rate performance \cite{Fan_15, Fan_16, Yu_18, Zeina_19, Chen_20, Zhang_Di_21, Liu_22_Power-Imbalanced, Wen_22}. A common approach is to first construct a multi-dimensional mother constellation (MC), with which sparse codebooks are generated through certain operations over the MC such as permutation, interleaving, and phase rotation. In general, the MC and the associated operations are designed to minimize the pairwise error probability (PEP), and metrics such as the minimum Euclidean distance (MED), minimum product distance (MPD), and diversity order are commonly considered as design criteria. A systematic MC design framework was first introduced in \cite{Taherzadeh_SCMA_first}. To maximize the MED, several MC designs based on star-QAM, constellation rotation, and interleaving were proposed in \cite{Fan_15, Fan_16, Yu_18}. Alternative MC design methods by taking advantage of the golden-angle modulation \cite{Zeina_19}, channel-specific near-optimal codebook design \cite{Chen_20}, uniquely decomposable constellation-group \cite{Zhang_Di_21}, and power imbalance across different codebooks \cite{Liu_22_Power-Imbalanced}, were subsequently developed. Additionally, enhanced MC designs targeting both MED and MPD were investigated in \cite{Wen_22}.

\subsection{Motivations and Contributions}

Although a number of SVC variants have been proposed in recent years, existing studies mostly focus on sparse vector mappings and corresponding decoding algorithms. From a channel coding perspective, the distance properties of the resulting codebooks have received little research attention, yet this is essential for understanding the SVC transmission reliability. With the aid of certain constant weight codes, the MED of single-section sparse regression codes (which is essentially SVC) was studied in \cite{Mow_24_CWC}. However, a systematic characterization of the pairwise Euclidean distance spectrum of the resulting codewords remains unexplored. Moreover, the impact of the underlying sparse vector structure on the SVC codebook distance properties has not been well understood. In contrast, as stated in the previous subsection, SCMA codebook design has been extensively studied from a distance-oriented perspective, whereby a primary objective is to design sparse codebooks with large MED. Since both SCMA and SVC are sparsity driven coding schemes, it is intriguing to understand how to exploit the existing know-hows of SCMA to further enhance the SVC design, and vice versa. For the former, specifically, how can we achieve point-to-point URLLC transmission by properly extracting the multiuser coding gain and the constellation shaping gain associated to SCMA? 

These motivate us to study SCMA-SVC for enhanced short-packet URLLC transmissions in this work. Our key idea is to jointly design the  sparse patterns and constellation symbols of SVC sparse vectors according to the sparse indicator matrix and multi-dimensional MC of SCMA, respectively. As a result, we show that the codeword distance properties of SCMA-SVC can be explicitly characterized. Furthermore, we investigate the diversity properties of SCMA-SVC over Rayleigh fading channels and how phase rotation applied to the nonzero entries of sparse vectors reshapes the codeword difference structure for achieving full diversity order. Simulation results demonstrate that the proposed SCMA-SVC achieves superior BLER performance compared with existing SVC schemes over both Gaussian and Rayleigh fading channels. The main contributions of this paper are summarized as follows:


\begin{itemize}
    
    \item We establish a fundamental relationship between the Euclidean distance spectrum of SVC codebooks and that of the underlying sparse vector set. In particular, we show that the pairwise codeword distances are fully determined by the pairwise distances among sparse vectors.
    
    \item We propose a novel SCMA-SVC scheme that jointly exploits sparse-pattern separation and symbol separation to construct sparse vector sets with analytically characterized distance properties. It is found that the resulting MED is determined by both the SCMA sparse matrix structure and the multi-dimensional constellation, yielding a systematic framework for improving codebook distance properties.

    \item We propose a low-complexity near-maximum-likelihood (ML) decoder that leverages the sparse pattern indexing mechanism of SCMA-SVC. By employing a generalized likelihood ratio test (GLRT) to identify a small set of candidate patterns, ML search can be performed over a substantially reduced candidate set, thereby significantly reducing the search complexity compared to the exhaustive ML detection while maintaining near-ML performance.
    
    
    
  

\end{itemize}

 




The remainder of this paper is organized as follows. We briefly introduce SVC and SCMA codebook design in Section~\ref{Sec:Introcuction to SCMA and SVC}. Section~\ref{Sec:SCMA_SVC} presents the proposed SCMA-SVC scheme and analyzes fundamental properties. Our low-complexity near-ML decoder and the corresponding complexity analysis are presented in Section~\ref{Sec:SCMA_SVC_decoder}. The performance comparisons are provided in Section~\ref{Sec:Simulation}. Finally, Section~\ref{Sec:Conclusion} ends this paper with the conclusion.

\subsection{Notations}
The notations used throughout this paper are as follows:
\begin{itemize}
    \item Lower-case letters denote scalars, boldface lower-case letters denote vectors,
and boldface upper-case letters denote matrices. Calligraphic letters denote sets. For example, $a$, $\bm a$, $\bm A$, and $\mathcal{A}$, respectively;
    \item $\bm A^{T}$, $\bm A^{H}$, and $\bm A^{-1}$ denote the transpose, Hermitian, and inverse of the matrix $\bm A$, respectively;
    \item $|\mathcal{A}|$ denotes the cardinality of the set $\mathcal{A}$;
    \item $\|\cdot\|_{p}$ denotes the $l_p$-norm;
    \item $\binom{N}{K}$ denotes the number of combinations of choosing $K$ out of $N$;
    \item $\left \lfloor \cdot \right\rfloor$ denotes the floor operation;
    \item $\bm{I}_N$ denotes the $N \times N$ identity matrix;
    \item $\text{diag}(\bm{A})$ denotes the vector composed of the diagonal elements of the matrix $\bm A$;
    \item $\mathcal{A} \setminus \mathcal{B}$ denotes the set difference, i.e., the set of elements in $\mathcal{A}$ but not in $\mathcal{B}$.
\end{itemize}

\section{Background and Preliminaries}\label{Sec:Introcuction to SCMA and SVC}

\subsection{Sparse Vector Coding}

The encoding process of the SVC scheme is depicted in Fig.~\ref{SVC_scheme}. In SVC, each information vector $\bm{m}$ is first mapped to a sparse vector $\bm{s}=(s_0,s_1,\ldots,s_{N-1})^{T}$ of length $N$ with $K$ nonzero entries. By selecting $K$ positions out of $N$, the number of distinct sparse patterns is $\binom{N}{K}$, which enables the transmission of at least
\begin{equation}
    \left \lfloor \log_2\binom{N}{K} \right\rfloor \geq b_t
\end{equation}
bits. After sparse mapping, the transmitted codeword $\bm{x}$ is generated by the codebook matrix. Specifically, let $\bm{C} = \begin{bmatrix} \bm{c}_{0},\bm{c}_{1},\cdots,\bm{c}_{N-1} \end{bmatrix}$ be the codebook matrix where $\bm{c}_n = (c_{0,n},c_{1,n},\ldots,c_{L-1,n})^{T}$ represents the $n$-th spreading sequence of length $L$. The transmitted codeword is then expressed as
\begin{equation}
\bm{x} = \bm{C}\bm{s} = \sum_{n=0}^{N-1} \bm{c}_n s_n.
\end{equation}
In general, the elements of the transmitted codeword form a QAM constellation of size $Q$. The normalization factor is given by $ \sqrt{\tfrac{2(Q-1)}{3}}$. Taking $b_t=3$, $N=5$, $K=2$, and $Q=4$ as an example, the sparse mapping rules and the corresponding codewords are demonstrated in Table~\ref{SVC_eg_mapping}. 


\begin{figure}[!t]
    \centering
    \begin{center}
        \includegraphics[width=0.5\textwidth]{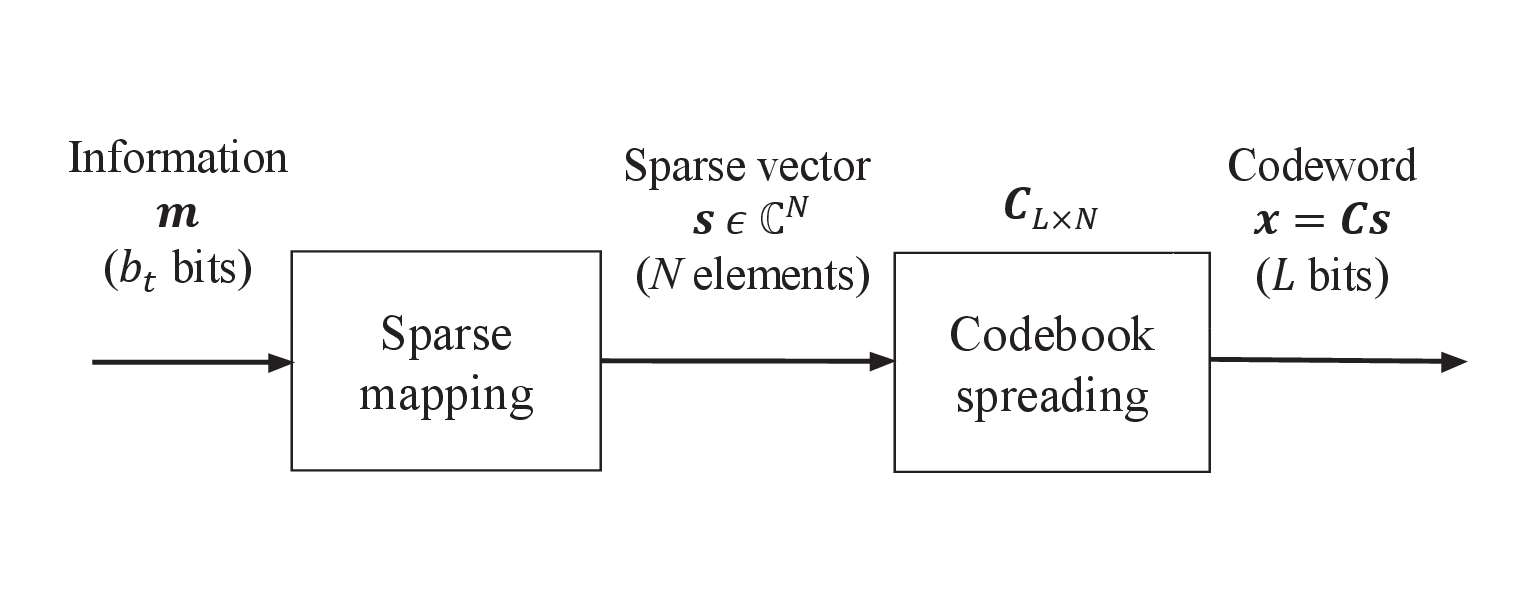}
        \caption{SVC encoding process.}\label{SVC_scheme}
    \end{center}
\end{figure}

\begin{table}[!t]
\centering
\caption{An Example of Sparse Mapping With $b_t=3$, $N=5$, $K=2$ and $Q=4$.\label{SVC_eg_mapping}}
\renewcommand{\arraystretch}{1.3}
\begin{tabular}{|c|c|c|} \hline
Information vector $\bm{m}$ & Sparse vector $\bm{s}$     &Codeword  $\bm{x}$ \\ \hline
$000$ & $(1/\sqrt{2})(1,j,0,0,0)^{T}$ & $(1/\sqrt{2})(\bm{c}_0+j\bm{c}_1)$ \\ \hline
$001$ & $(1/\sqrt{2})(1,0,j,0,0)^{T}$ & $(1/\sqrt{2})(\bm{c}_0+j\bm{c}_2)$ \\ \hline
$010$ & $(1/\sqrt{2})(0,1,j,0,0)^{T}$ & $(1/\sqrt{2})(\bm{c}_1+j\bm{c}_2)$ \\ \hline
$011$ & $(1/\sqrt{2})(1,0,0,j,0)^{T}$ & $(1/\sqrt{2})(\bm{c}_0+j\bm{c}_3)$ \\ \hline
$100$ & $(1/\sqrt{2})(0,1,0,j,0)^{T}$ & $(1/\sqrt{2})(\bm{c}_1+j\bm{c}_3)$ \\ \hline
$101$ & $(1/\sqrt{2})(0,0,1,j,0)^{T}$ & $(1/\sqrt{2})(\bm{c}_2+j\bm{c}_3)$ \\ \hline
$110$ & $(1/\sqrt{2})(1,0,0,0,j)^{T}$ & $(1/\sqrt{2})(\bm{c}_0+j\bm{c}_4)$ \\ \hline
$111$ & $(1/\sqrt{2})(0,1,0,0,j)^{T}$ & $(1/\sqrt{2})(\bm{c}_1+j\bm{c}_4)$ \\ \hline
\end{tabular}
\end{table}

\subsection{SCMA Codebooks}

Consider an SCMA system with $J$ users over $N$ orthogonal resource nodes (e.g., subcarriers in a multicarrier system), whereby $J>N$ so that the system needs to support an overloading factor of $J/N$. Each user (say, User $j$, where $0\leq j \leq J-1$) is assigned a sparse codebook 
\begin{equation}
    \bm{S}_j=\begin{bmatrix} \bm{s}_{j,0},\bm{s}_{j,1},\cdots,\bm{s}_{j,Q-1} \end{bmatrix},
\end{equation}
consisting of $Q$ sparse codewords of length $N$, where each codeword contains only $K$ nonzero entries, satisfying $K < N$.


The sparse structure of the SCMA codebooks across the $J$ users can be characterized by a sparse indicator matrix 
\begin{equation}
\bm{F} = \begin{bmatrix} \bm{f}_{0},\bm{f}_{1},\cdots,\bm{f}_{J-1} \end{bmatrix},
\end{equation}
where each column $\bm{f}_j$ is obtained from the binary mapping matrix $\bm {V}_{j}$ of size $N \times K$, i.e., 
\begin{equation}
    \bm{f}_j=\text{diag}(\bm {V}_{j}\bm {V}_{j}^{T}).
\end{equation}
Similar to the parity check matrix in a low density parity check code, the sparse indicator matrix is typically designed to avoid short cycles, particularly cycles of length four, in the corresponding Tanner graph. As a result, any two columns overlap in at most one position, thereby improving the separation among different sparse patterns. A sparse indicator matrix example for $J=6$, $N=4$, and $K=2$ is given by
   \begin{equation} \bm{F}= \begin{bmatrix} 1&1&1&0&0&0\\ 1&0&0&1&1&0\\ 0&1&0&1&0&1\\ 0&0&1&0&1&1\\ \end{bmatrix}.\end{equation}
The corresponding binary mapping matrices are
   \begin{equation}\begin{aligned}\bm {V}_{0}=&\begin{bmatrix} 1 & 0\\ 0&1\\ 0 & 0\\ 0 &0\\ \end{bmatrix},~~\bm {V}_{1} = \begin{bmatrix} 1 & 0\\ 0 & 0\\ 0 & 1\\ 0 & 0\\ \end{bmatrix},~\bm {V}_{2} = \begin{bmatrix} 1 & 0\\ 0 & 0\\ 0 & 0\\ 0 & 1\\ \end{bmatrix}, \\ \bm {V}_{3}=&\begin{bmatrix} 0 & 0\\ 1 & 0\\ 0&1\\ 0 &0\\ \end{bmatrix},~~\bm {V}_{4} = \begin{bmatrix} 0 & 0\\ 1 & 0\\ 0 & 0\\ 0&1\\ \end{bmatrix},~\bm {V}_{5} = \begin{bmatrix} 0&0\\ 0 & 0\\ 1 &0\\ 0&1\\ \end{bmatrix}.\end{aligned}\end{equation}
Each binary mapping matrix $\bm{V}_j$ specifies the locations of the nonzero entries associated with the sparse pattern $\bm{f}_j$.

For the $j$-th user, the codebook can be expressed as
\begin{equation} 
\bm{S}_{j}=\bm {V}_{j} \bm{\Delta}_{j} \bm{A}_{\text {MC}},~~\text {for}~j=0,1,\cdots ,J-1,
\end{equation}
where $\bm{\Delta}_{j}$ denotes its constellation operator, e.g., a permutation or rotation applied to the multi-dimensional constellation $\bm{A}_{\text{MC}}=\begin{bmatrix}\bm{a}_0,\bm{a}_1,\ldots,\bm{a}_{Q-1}\end{bmatrix}$, in which $\bm{a}_q \in \mathbb{C}^K$ denotes the $q$-th constellation point of~$\bm{A}_{\text{MC}}$. As an illustrative example, for $K=3$ and $Q=4$, the MC is given in \eqref{A_MC_3_4} where each column represents a constellation point. 

\begin{figure*}[!t]
\begin{equation}\label{A_MC_3_4}
\bm{A}_{\text{MC}} =\begin{bmatrix}\bm{a}_0,\bm{a}_1,\bm{a}_2,\bm{a}_{3}\end{bmatrix}=
\begin{bmatrix}
0.5381 + 0.2092i &  0.2092 - 0.5381i & -0.5381 - 0.2092i & -0.2092 + 0.5381i \\
-0.5597 - 0.1417i &  0.5597 + 0.1417i &  0.1417 - 0.5597i & -0.1417 + 0.5597i \\
-0.3649 - 0.4474i & -0.4474 + 0.3649i &  0.4474 - 0.3649i  & 0.3649 + 0.4474i 
\end{bmatrix}.
\end{equation}
\end{figure*}

Therefore, the transmitted sparse codeword for the $j$-th user can be expressed as 
\begin{equation} 
\bm{s}_{j,q}=\bm {V}_{j} \bm{\Delta}_{j} \bm{a}_q,~~\text {for}~j=0,1,\cdots ,J-1,
\end{equation}
where $\bm{a}_q$ is selected according to $\log_2(Q)$ input bits of the $j$-th user. It is noted that the sparse mapping matrix determines the sparse-pattern structure of the resulting SCMA codebook, whereas the multi-dimensional MC determines the symbol values associated with each sparse pattern.

\section{SCMA Inspired SVC (SCMA-SVC)}\label{Sec:SCMA_SVC}

In this section, we introduce the proposed SCMA-SVC scheme. The basic idea is to incorporate the sparse codebook design principles of SCMA to generate new SVC codebooks with large MED and full diversity order. 

As shown in Fig.~\ref{SVC_scheme}, each transmitted codeword in SVC can be viewed as a linear mapping from the sparse vector through the codebook matrix. We first establish a relationship between the MED of the sparse vector set and that of the transmitted codebook. 


Let $\mathcal{S}=\{\bm{s}_0,\bm{s}_1,\ldots,\bm{s}_{M-1}\}$ denote a sparse vector set, where each sparse vector $\bm{s}_l=(s_{l,0},s_{l,1},\ldots,s_{l,N-1})^{T}\in \mathbb{C}^{N}$ contains $K$ nonzero entries.

\begin{defin}
The minimum Euclidean distance of a set $\mathcal{S}=\{\bm{s}_0,\bm{s}_1,\ldots,\bm{s}_{M-1}\}$ is defined as 
\begin{equation}
    d_{\min}(\mathcal{S}) \triangleq \min_{\substack{\bm{s}_l,\bm{s}_p\in\mathcal{S},\\ l\neq p}}
\|\bm{s}_l-\bm{s}_p\|_2.
\end{equation}
\end{defin}


\begin{thm} \label{dmin_theorem}
Consider a finite sparse vector set 
\begin{equation}
\mathcal{S}=\{\bm{s}_0,\bm{s}_1,\ldots,\bm{s}_{M-1}\} \subset \mathbb{C}^{N} 
\end{equation}
with minimum Euclidean distance $d_{\min}(\mathcal{S})$. Let
\begin{equation}
\mathcal{X}=\{\bm{x}_l=\bm C\bm{s}_l \mid l=0,1,\ldots,M-1\},
\end{equation}
where $\bm C \in \mathbb{C}^{L\times N}$ satisfies $\bm{C}^{H}\bm{C}=\alpha \bm{I}_{N}$. Then, the minimum Euclidean distance of $\mathcal{X}$ satisfies
\begin{equation}
d_{\min}(\mathcal{X})
=
\sqrt{\alpha}\,d_{\min}(\mathcal{S}).
\end{equation}
\end{thm}

\begin{IEEEproof}
For any $l \neq p$, the pairwise Euclidean distance between $\bm{x}_l$ and $\bm{x}_p$ is
\begin{equation}
    \begin{aligned}
    \|\bm{x}_l - \bm{x}_p\|_2
    &= \|\bm{C}\bm{s}_l - \bm{C}\bm{s}_p\|_2 = \|\bm{C}\left(\bm{s}_l -\bm{s}_p\right)\|_2\\
    &= \sqrt{\big(\bm{C}\left(\bm{s}_l -\bm{s}_p\right)\big)^{H}\bm{C}\big(\bm{s}_l - \bm{s}_p\big)} \\
    &= \sqrt{(\bm{s}_l - \bm{s}_p)^{ H}\bm{C}^{ H}\bm{C}(\bm{s}_l - \bm{s}_p)}.
    \end{aligned}
\end{equation}
Since $\bm{C}$ satisfies $\bm{C}^{\mathrm H}\bm{C} = \alpha\bm{I}_N$, it simplifies to
\begin{equation}\label{Codewords_distance}
    \|\bm{x}_l - \bm{x}_p\|_2
    = \sqrt{\alpha}\|\bm{s}_l - \bm{s}_p\|_2.
\end{equation}
Taking the minimum over all $l\neq p$ yields
\begin{equation}
d_{\min}(\mathcal{X})=\sqrt{\alpha}\,d_{\min}(\mathcal{S}).
\end{equation}
\end{IEEEproof}

\begin{crly} \label{phase_roation_corollary}
Under the assumptions of {\it Theorem~\ref{dmin_theorem}}, let $\bm{P}$ be a diagonal phase rotation matrix given by
\begin{equation}\label{P_form}
\bm{P} =
\begin{bmatrix}
e^{j\phi_0} & 0 & \cdots & 0 \\
0 & e^{j\phi_1} & \cdots & 0 \\
\vdots & \vdots & \ddots & \vdots \\
0 & 0 & \cdots & e^{j\phi_{N-1}}
\end{bmatrix},
\end{equation}
where $|e^{j\phi_n}|=1$ for all $n=0,1,\ldots,N-1$. Define
\begin{equation}
\mathcal{X}=\{\bm{x}_l= \bm C \bm P \bm{s}_l \mid l=0,1,\ldots,M-1\}.
\end{equation}
Then, the minimum Euclidean distance of $\mathcal{X}$ satisfies
\begin{equation}
d_{\min}(\mathcal{X})
=
\sqrt{\alpha}\,d_{\min}(\mathcal{S}).
\end{equation}
\end{crly}

\begin{IEEEproof} Let $\bm{B}=\bm{C}\bm{P}$. Since $\bm{P}$ is unitary,
$\bm{B}^{H}\bm{B}=\alpha\bm{I}_{N}$. Therefore, the result follows from {\it Theorem~\ref{dmin_theorem}}.
\end{IEEEproof}

\begin{figure*}[!t]
    \centering
    \begin{center}
        \includegraphics[width=0.88\textwidth]{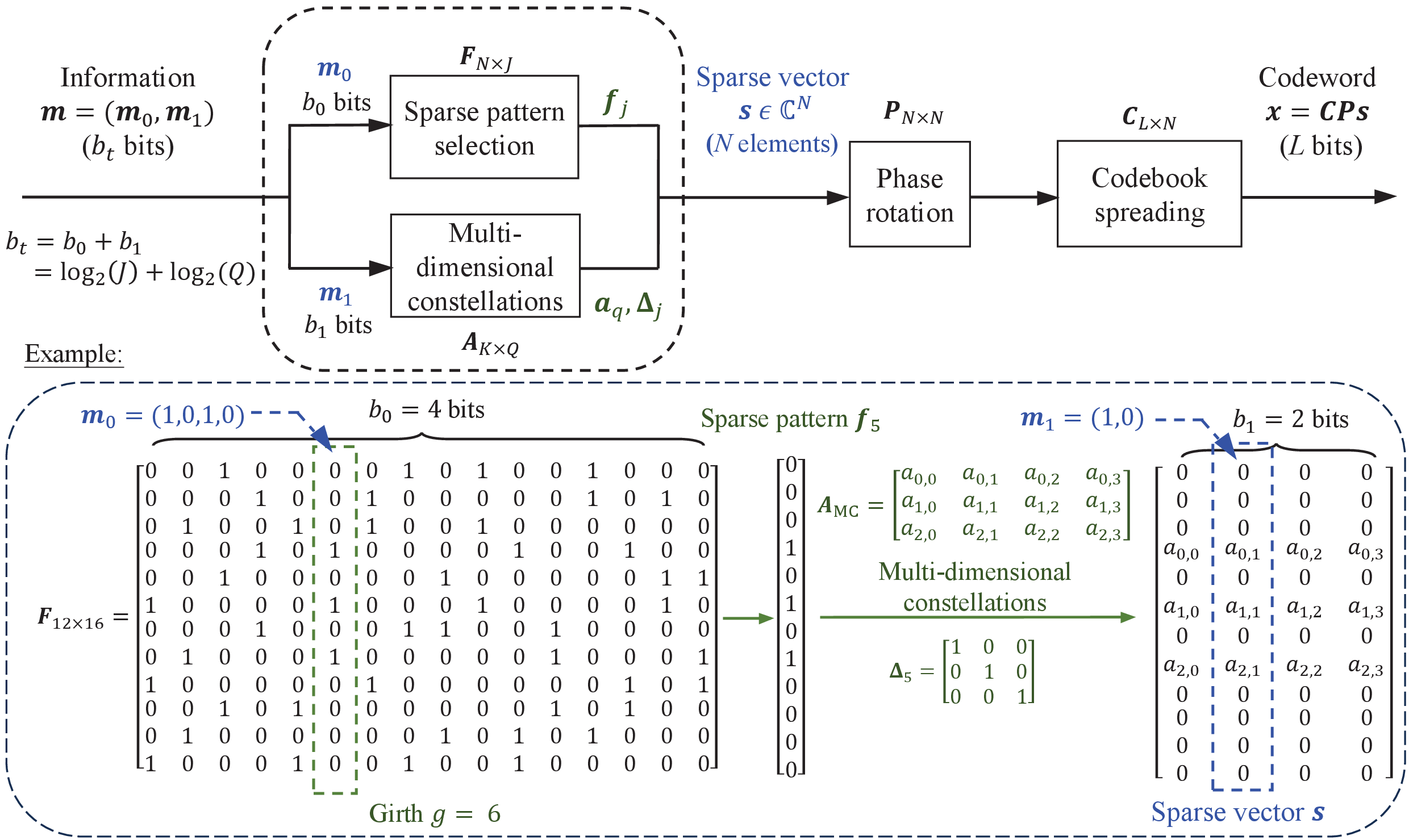}
        \caption{The encoding process of SCMA-SVC.}\label{SCMA_SVC_Scheme}
    \end{center}
\end{figure*}

\begin{remark} \label{main_design_1}
According to {\it Theorem~\ref{dmin_theorem}}, when the sparse vector set $\mathcal{S}$ has an MED 
$d_{\min}(\mathcal{S})=d_S$ and an orthogonal codebook matrix $\bm C$ is employed, the resulting transmitted codebook~$\mathcal{X}$ achieves an MED
\begin{equation}\label{Remark_1_eq}
d_{\min}(\mathcal{X})
=
\sqrt{\alpha} \cdot d_S.
\end{equation}
Thus, the MED of the transmitted codebook $\mathcal{X}$ is directly determined by that of the sparse vector set $\mathcal{S}$. Furthermore, \eqref{Codewords_distance} implies that all pairwise Euclidean distances in the transmitted codebook are scaled by a factor of $\sqrt{\alpha}$ from those in the sparse vector set. Therefore, from a coding perspective, maximizing the pairwise codeword distances is equivalent to increasing the pairwise distances among the sparse vectors in~$\mathcal{S}$. Consequently, the Euclidean distance spectrum of the transmitted codebook can be characterized through sparse vector design.
\end{remark}

\begin{remark} \label{main_design_2}
According to {\it Corollary~\ref{phase_roation_corollary}}, inserting a diagonal phase rotation matrix $\bm P$ before the codebook matrix $\bm C$ does not affect the MED scaling. Hence, the essential distance characteristics of the transmitted codebook are dictated by the sparse vector set $\mathcal{S}$. Moreover, the phase rotation reshapes the structure of the transmitted codeword differences. As will be shown in Section~\ref{subsec:pep_kpi}, a properly designed phase rotation can enable the resulting codewords to achieve full diversity order under Rayleigh fading channels.
\end{remark}

Notably, in SCMA systems, each user-specific codeword is inherently sparse as determined by the corresponding sparse indicator matrix. Moreover, the multi-dimensional MC is typically designed to achieve a large MED. These properties are consistent with the objective of maximizing MED. Therefore, we propose the SCMA-SVC scheme by leveraging both the sparse pattern design principles and the underlying of multi-dimensional constellations to construct a sparse vector set $\mathcal{S}$ whose MED can be explicitly characterized.

\subsection{SCMA-SVC Encoding}\label{sub_sec_SCMA_SVC}
Let us introduce the SCMA-SVC encoder in Fig.~\ref{SCMA_SVC_Scheme}. The encoder maps the information $\bm{m}$ of $b_t$ bits into a complex sparse vector $\bm{s}=(s_0,s_1,\ldots,s_{N-1})^{T}$ with $K$ nonzero entries. The information vector $\bm{m}$ is split into two vectors of lengths $b_0$ and $b_1$, denoted by $\bm{m}_0=(m_{0,0},m_{0,1},\ldots,m_{0,b_0-1})$ and $\bm{m}_1=(m_{1,0},m_{1,1},\ldots,m_{1,b_1-1})$, respectively. Then, $\bm{m}_0$ is used to select a sparse pattern $\bm f_j$ from the SCMA indicator matrix $\bm{F}=\begin{bmatrix}\bm f_0,\bm f_1,\ldots,\bm f_{J-1}\end{bmatrix}$, where $j$ is the decimal representation of the binary vector $\bm{m}_0$, and each column of~$\bm{F}$ specifies a sparse pattern with $K$ nonzero elements, i.e., $\|\bm{f}_j\|_{0}=K$. On the other hand, $\bm{m}_1$ is used to select a multi-dimensional constellation symbol $\bm{a}_q$ from the MC matrix $\bm{A}_{\text{MC}}=\begin{bmatrix}\bm{a}_0,\bm{a}_1,\ldots,\bm{a}_{Q-1}\end{bmatrix}$, where $q$ is the decimal representation of the binary vector $\bm{m}_1$. For $j=0,1,\ldots,J-1$ and $q=0,1,\ldots,Q-1$, the resultant sparse vector is obtained as 
\begin{equation}\label{eq:mapped_sparse_vector}
\bm{s}_l=\bm{s}_{jQ+q}=\bm {V}_{j} \bm{\Delta}_{j} \bm{a}_q,
\end{equation}
where $\bm{\Delta}_{j}$ represents the constellation operator corresponding to the $j$-th sparse pattern. Consequently, the sparse vector set $\mathcal{S}$ is given by
\begin{equation}
\mathcal{S}=\{\bm{s}_{l} \mid l=0,1,\ldots,M-1\},
\end{equation}
where $M=JQ$ and $\|\bm{s}_l\|_{2}^{2}=1$ for all $l$. Next, each sparse vector $\bm s$ is processed by a diagonal phase rotation matrix $\bm{P}$, and then spread by certain column vectors of the codebook matrix $\bm{C}$ to form the transmitted codeword
\begin{equation}
    \bm{x} =\bm{C} \bm{P} \bm{s},
\end{equation}
where $\bm{P}$ is defined in \eqref{P_form} with $|e^{j\phi_n}| = 1$ for all $n$. 

Since the sparse-pattern index $j$ is selected from $J$ candidate sparse patterns and the constellation index $q$ is selected from $Q$ constellation points, the proposed SCMA-SVC supports $M=JQ$ distinct sparse vectors. Therefore, the total number of information bits carried by each codeword is 
\begin{equation}
b_t=\log_2(M)=\log_2(J)+\log_2(Q).
\end{equation}
For a codeword length $L$, the corresponding code rate of the proposed SCMA-SVC is given by
\begin{equation}
R=\frac{b_t}{L}=\frac{\log_2(J)+\log_2(Q)}{L}.
\end{equation}

According to {\it Theorem~\ref{dmin_theorem}} and {\it Corollary~\ref{phase_roation_corollary}}, when an orthogonal codebook matrix $\bm{C}$ is employed, the MED of the transmitted codebook $\mathcal{X}=\{\bm{x}_{l}=\bm{C}\bm{P}\bm{s}_l \mid l=0,1,\ldots,M-1\}$ is determined by that of the sparse vector set $\mathcal{S}$. To increase the pairwise Euclidean distances among the sparse vectors in $\mathcal{S}$, we propose the following design principles based on sparse pattern separation and symbol separation.

{\it Design criterion:}

1) Sparse pattern separation: The sparse indicator matrix $\bm{F}$ is designed with girth $g\ge 6$\text{\,\begingroup\renewcommand{\thefootnote}{\fnsymbol{footnote}}\footnotemark[2]\endgroup}, to ensure that any two columns $\bm{f}_l$ and $\bm{f}_p$ overlap in at most one position, thereby increasing the separation between the corresponding sparse patterns.

2) Symbol separation: The multi-dimensional mother constellation $\bm{A}_{\text{MC}}$, together with the constellation operator $\bm{\Delta}_{j}$, is designed to maximize the MED among the symbols within each sparse pattern.

\begingroup
\renewcommand{\thefootnote}{\fnsymbol{footnote}}
\footnotetext[2]{The girth of a Tanner graph, i.e., the bipartite graph representation of a sparse matrix, is defined as the length of the shortest cycle in the graph. A girth of at least $6$ implies that the graph contains no cycles of length $4$.}
\endgroup


By combining a large-girth sparse indicator matrix with a proper multi-dimensional MC, a sparse vector set $\mathcal{S}$ with a large MED can be constructed. 

Next, we analyze the MED of the sparse vector set generated by the proposed SCMA-SVC scheme. It is emphasized that each sparse vector $\bm{s} \in \mathcal{S}$ contains $K$ nonzero elements, whose positions are determined by the indicator matrix $\bm{F}$. The values of these nonzero elements are drawn from the $\bm{A}_{\text{MC}}$. Due to the use of a large-girth indicator matrix, any two columns of~$\bm{F}$ have at most one common nonzero position. As a result, the MED of the sparse vector set is determined by the following two cases.

{\it Case 1:} Suppose that $\bm{s}_l$ and $\bm{s}_p$ are two sparse vectors whose nonzero positions are distinct. Let $\mathcal{I}_l$ and $\mathcal{I}_p$ denote the index sets of their nonzero positions. Owing to the large-girth property of $\bm{F}$, the overlap between the two sparse patterns satisfies
\begin{equation}
|\mathcal{I}_l \cap \mathcal{I}_p| \le 1.
\end{equation}
The squared Euclidean distance between $\bm{s}_l$ and $\bm{s}_p$ can therefore be expressed as
\begin{equation}
\|\bm{s}_l-\bm{s}_p\|_{2}^{2}
=
\sum_{i\in \mathcal{I}_l\setminus \mathcal{I}_p}|a_i|^2
+
\sum_{i\in \mathcal{I}_p\setminus \mathcal{I}_l}|b_i|^2
+
\sum_{i\in \mathcal{I}_l\cap \mathcal{I}_p}|a_i-b_i|^2,
\end{equation}
where $a_i$ and $b_i$ denote the constellation symbols drawn from~$\bm{A}_{\text{MC}}$, and each constellation symbol has energy $1/K$, i.e., $|a_i|^2 = |b_i|^2 = 1/K$. We consider a worst-case scenario where two sparse vectors share exactly one overlapping position and the constellation symbols at the overlapping position are identical. In this case, the contribution of the overlapping position becomes zero. Since each constellation symbol has energy $1/K$, the remaining $K-1$ non-overlapping positions contribute a total squared distance of $2(K-1)/K$. Therefore, the Euclidean distance becomes
\begin{equation}\label{worst_case_eq}
d_{1}=\|\bm{s}_l-\bm{s}_p\|_{2}= \sqrt{\frac{2(K-1)}{K}}.
\end{equation}
This indicates that the MED in {\it Case~1} is determined by the $K-1$ non-overlapping positions of the sparse vectors together with the symbol energy $1/K$.

{\it Case 2:} We next consider two sparse vectors $\bm{s}_l$ and $\bm{s}_p$ that share the same sparse pattern, i.e., $\mathcal{I}_l=\mathcal{I}_p$. The Euclidean distance depends on the difference between the constellation symbols, i.e.,  $\|\bm{s}_l-\bm{s}_p\|_2
=\|\bm{a}_l-\bm{a}_p\|_2$ where $\bm{a}_l$ and $\bm{a}_p$ denote two columns of the multi-dimensional constellation $\bm{A}_{\text{MC}}$. Therefore, the MED in this case is determined by the MED of $\bm{A}_{\text{MC}}$. That is,
\begin{equation}
d_{2}=d_{\min}(\mathcal{A}_{\text{MC}}).
\end{equation}

Combining the above two cases, the MED of the sparse vector set $\mathcal{S}$ generated by the proposed SCMA-SVC can be expressed as
\begin{equation}\label{MED_of_sparse_vector}
d_{\min}(\mathcal{S})
=\min\left\{d_{1}, d_{2}\right\}=\min\left\{\hspace{-0.2em}\sqrt{\frac{2(K-1)}{K}}, d_{\min}(\mathcal{A}_{\text{MC}})\hspace{-0.2em}\right\}.
\end{equation}
This result reveals that the MED of the sparse vector set is jointly determined by two factors: the sparse pattern structure induced by the indicator matrix~$\bm{F}$ and the MED of the multi-dimensional constellation $\bm{A}_{\text{MC}}$. Therefore, improving the overall MED requires both a large-girth indicator matrix and a mother multi-dimensional constellation  with a large MED. The design of $\bm{A}_{\text{MC}}$ may follow the distance-optimized SCMA codebook design commonly adopted in the literature \cite{Beko_12, Wen_22}.

\begin{remark}\label{Remark_SCMA_dim_form}
By substituting \eqref{MED_of_sparse_vector} into \eqref{Remark_1_eq}, the MED of the resulting transmitted codebook $\mathcal{X}$ can be obtained as
\begin{equation}\label{eq_SCMA_dim_form}
d_{\min}(\mathcal{X})=\sqrt{\alpha}\min\left\{\sqrt{\frac{2(K-1)}{K}}, d_{\min}(\mathcal{A}_{\text{MC}})\right\}.
\end{equation}
\eqref{eq_SCMA_dim_form} provides an explicit characterization of the MED of the transmitted codebook.
\end{remark}

\begin{eg}\label{eg_SCMA-SVC_6_bits}
As shown in Fig.~\ref{SCMA_SVC_Scheme}, we consider an SCMA-SVC encoder that encodes $b_t=6$ bits. The indicator matrix $\bm{F}_{12 \times 16}$ shown in Fig.~\ref{SCMA_SVC_Scheme} is adopted, which has girth $6$. The $3 \times 4$ multi-dimensional constellation $\bm{A}_{\text{MC}}$ shown in (\ref{A_MC_3_4}), obtained from \cite{Wen_22} is employed, where $d_{\min}(\mathcal{A_{\text{MC}}})=1.633$. A $16 \times 12$ codebook matrix $\bm{C}$ is chosen as a bipolar orthogonal matrix satisfying $\bm C^{H}\bm C=16 \bm I_{12}$. For illustration purpose, the phase rotation matrix $\bm{P}=\bm{I}_{12}$ (i.e., $\phi_n=0$ for all $n$) is used. We have $b_0=4$, $b_1=2$, $N=12$, $L=16$, $K=3$, $Q=4$, and $\alpha=16$. The first $b_0=4$ bits determine the sparse pattern index $j$, and the remaining $b_1=2$ bits determine the constellation index $q$. Assume that the transmitted information bits are $\bm{m}=(1,0,1,0,1,0)$. Then, we have $\bm{m}_0=(1,0,1,0)$ and $\bm{m}_1=(1,0)$, resulting in $j=5$ and $q=1$. Hence, the sparse pattern $\bm{f}_5$ is selected from the indicator matrix $\bm{F}$ with the corresponding binary mapping matrix $\bm{V}_5$. The selected constellation vector is $\bm{a}_1$ from $\bm{A}_{\text{MC}}$ with the corresponding constellation operator $\bm{\Delta}_5=\bm{I}_{3}$. According to \eqref{eq:mapped_sparse_vector}, the mapped sparse vector is given by 
\begin{equation}
\begin{aligned}
    \bm s_{21}&=\bm{V}_5\bm{\Delta}_5\bm{a}_1=(0,0,0,0.2092 - 0.5381i,0,\\&0.5597 + 0.1417i,0,-0.4474+ 0.3649i,0,0,0,0)^{T},
\end{aligned}
\end{equation}
which contains $3$ nonzero entries. According to \eqref{MED_of_sparse_vector}, the MED of the resulting sparse vector set $\mathcal{S}$ can be obtained as $d_{\min}(\mathcal{S})=\min\{1.155,1.633\}=1.155$. Then, the corresponding codebook $\mathcal{X}$ can be generated with an MED $d_{\min}(\mathcal{X})=\sqrt{16} \times 1.155 = 4.62$.
\end{eg}

\subsection{KPIs of the Proposed SCMA-SVC Codebooks Under Gaussian and Rayleigh Channels}\label{subsec:pep_kpi}

In this subsection, we present the key performance indicators (KPIs) of the proposed SCMA-SVC codebooks including MED, diversity order, and MPD under both Gaussian and Rayleigh fading channels. In addition, we design the phase rotation matrix $\bm{P}$ to enable the proposed SCMA-SVC to achieve full diversity order.

At the receiver, the received signal $\bm{y}$ after the fading channel is given by
\begin{equation}
    \bm y= \bm{H} \bm x + \bm w = \bm{H}\bm C \bm P \bm s + \bm w,
\end{equation}
where $\bm{H}$ is a diagonal matrix formed from the channel vector $\bm{h}=(h_0,h_1,\ldots, h_{L-1})$ and $\bm w=(w_0,w_1,\ldots, w_{L-1})$ is the complex additive white Gaussian noise, i.e., $w_i \sim \mathcal{CN}(0, \sigma^2)$. 

\subsubsection{PEP and BLER over Gaussian channels}

The pairwise error probability of any two distinct codewords $\bm{x}_0$ and $\bm{x}_1$ in $\mathcal{X}$, denoted by $P(\bm{x}_0\rightarrow \bm{x}_1)$, is defined as the probability that the decoder selects $\bm{x}_1$ when $\bm{x}_0$ is transmitted. Under ML decoding over the Gaussian channel, the PEP associated with the proposed SCMA-SVC scheme can be expressed as
\begin{equation}
\begin{aligned}
    &P(\bm{x}_0\rightarrow \bm{x}_1)=Q\!\left(\sqrt{\frac{\|\bm{x}_0-\bm{x}_1\|_{2}^{2}}{2N_0}}\right)\\&=Q\!\left(\sqrt{\frac{\|\bm{C}\bm{P}(\bm{s}_0-\bm{s}_1)\|_{2}^{2}}{2N_0}}\right)=Q\!\left(\sqrt{\frac{\alpha\|\bm{s}_0-\bm{s}_1\|_{2}^{2}}{2N_0}}\right),
\end{aligned}
\label{eq:pep_awgn}
\end{equation}
since $\bm{C}^{\mathrm H}\bm{C} = \alpha\bm{I}_{N}$ and $\bm{P}^H \bm{P} = \bm{I}_{N}$. Based on the pairwise error expression in (\ref{eq:pep_awgn}), assuming that all codewords are transmitted with equal probability, the BLER can be upper bounded by applying the union bound over all distinct codeword pairs in $\mathcal{X}$, leading to
\begin{equation}
\text{BLER}_{\text{SCMA-SVC}}
\le
\frac{1}{M}
\sum_{l=0}^{M-1}
\sum_{\substack{p=0 \\ p \ne l}}^{M-1}
P(\bm{x}_l \rightarrow \bm{x}_p).
\label{eq:BLER_union}
\end{equation}
\eqref{eq:BLER_union} shows that the BLER performance under the Gaussian channel is characterized by the Euclidean distance spectrum of the transmitted codebook $\mathcal{X}$. Hence, the MED is considered the primary KPI of the proposed SCMA-SVC scheme over Gaussian channels. 


\subsubsection{PEP and BLER over Rayleigh fading channels}\label{sec:full_DO} For any two distinct codewords $\bm{x}_0$ and $\bm{x}_1$, we define $\Delta_i = x_{0,i} - x_{1,i}$ as the $i$-th component-wise difference between $\bm{x}_0$ and $\bm{x}_1$, where $\bm{x}_0 = \bm{C}\bm{P}\bm{s}_0$ and 
$\bm{x}_1 = \bm{C}\bm{P}\bm{s}_1$.  The PEP conditioned on the channel matrix $\bm{H}$ is given by
\begin{equation}
\begin{aligned}
&P(\bm{x}_0\rightarrow \bm{x}_1\mid \bm{H})\\
&=
Q\!\left(\sqrt{\frac{\|\bm{H}\left(\bm{x}_0-\bm{x}_1\right)\|_2^2}{2N_0}}\right)=
Q\!\left(
\sqrt{\frac{\sum_{i=0}^{L-1}|h_i|^2\,|\Delta_i|^2}{2N_0}}
\right).
\end{aligned}
\label{eq:pep_cond_rayleigh}
\end{equation}
Then, by using the approximation of the Q-function \cite{Kim_08}
\begin{equation}
Q(x)\approx \frac{1}{12}e^{-x^2/2}+\frac{1}{6}e^{-2x^2/3},
\label{eq:Q_approx}
\end{equation}
and taking the expectation over the channel matrix $\bm{H}$, the PEP can be derived as follows \cite{Liu_21}:
\begin{equation}
\begin{aligned}
    &P(\bm{x}_0\rightarrow \bm{x}_1)\\
&\approx
\frac{1}{12}\prod_{i=0}^{L-1}\left(1+\frac{|\Delta_i|^2}{4N_0}\right)^{-1}
+
\frac{1}{6}\prod_{i=0}^{L-1}\left(1+\frac{|\Delta_i|^2}{3N_0}\right)^{-1}.
\end{aligned}
\label{eq:pep_closed_form}
\end{equation}

For any two distinct codewords $\bm{x}_0$ and $\bm{x}_1$, we define the index set of nonzero component-wise distances as
\begin{equation}
\mathcal{D}^{(\bm{x}_0,\bm{x}_1)}\triangleq\{i \mid \Delta_i \neq 0,\ 0\le i\le L-1\}
\label{eq:D_set_1}
\end{equation}
and the cardinality of set 
\begin{equation}
G_d^{(\bm{x}_0,\bm{x}_1)}\triangleq|\mathcal{D}^{(\bm{x}_0,\bm{x}_1)}|.
\label{eq:D_set_2}
\end{equation}
In high-SNR region, it follows that $1+\frac{|\Delta_i|^2}{4N_0}\approx\frac{|\Delta_i|^2}{4N_0}$ and $1+\frac{|\Delta_i|^2}{3N_0}\approx\frac{|\Delta_i|^2}{3N_0}$ as $N_0 \rightarrow 0$. Accordingly, the PEP can be written as
\begin{equation}
\begin{aligned}
&P(\bm{x}_0\to\bm{x}_1)\\
&\hspace{-0.2em}\approx\hspace{-0.3em}
\left(\frac{1}{N_0}\right)^{-G_d^{(\bm{x}_0,\bm{x}_1)}}
\hspace{-0.8em}\left(
\frac{4^{-G_d^{(\bm{x}_0,\bm{x}_1)}}}{12}
\hspace{-0.2em}+\hspace{-0.2em}
\frac{3^{-G_d^{(\bm{x}_0,\bm{x}_1)}}}{6}
\right)
\hspace{-0.5em}\prod_{i\in\mathcal{D}^{(\bm{x}_{0},\bm{x}_{1})}}\hspace{-0.3em}|\Delta_i|^{-2}.
\label{eq:pep_high_snr}
\end{aligned}
\end{equation}
Substituting \eqref{eq:pep_high_snr} into the union bound expression in \eqref{eq:BLER_union}, we obtain
\begin{equation}
\mathrm{BLER}_{\text{SCMA-SVC}}
\hspace{-0.1em}\lesssim\hspace{-0.1em}
\frac{1}{M}
\sum_{l=0}^{M-1}
\sum_{p\ne l}
U_{l,p}
\left(\frac{1}{N_0}\right)^{\hspace{-0.4em}-G_d^{(\bm{x}_l,\bm{x}_p)}}
\hspace{-3.3em}\prod_{i\in\mathcal{D}^{(\bm{x}_l,\bm{x}_p)}} |\Delta_i|^{-2},
\end{equation}
where $U_{l,p}$ is a constant independent of $N_0$. In the high-SNR regime, the BLER is dominated by the smallest achievable diversity order among all distinct codeword pairs. 
Accordingly, the diversity order of the codebook $\mathcal{X}$ is defined as
\begin{equation}
G_d 
\triangleq 
\min_{l \ne p} 
G_d^{(\bm{x}_l,\bm{x}_p)}.
\end{equation}

Among the codeword pairs achieving $G_d$, the performance is further governed by the corresponding MPD, defined as
\begin{equation}
\mathrm{MPD}^{(G_d)}
\triangleq
\min_{\substack{l \ne p \\
G_d^{(\bm{x}_l,\bm{x}_p)} = G_d}}
\prod_{i \in \mathcal{D}^{(\bm{x}_l,\bm{x}_p)}} 
|\Delta_i|.
\label{eq:MPD_def}
\end{equation}

Finally, combining the union bound with the PEP characterization reveals that the BLER over Rayleigh fading is asymptotically governed
by the diversity order and the corresponding MPD. Therefore, under Rayleigh fading channels, the KPIs of the proposed SCMA-SVC scheme are the diversity order $G_d$ and the associated $\mathrm{MPD}^{(G_d)}$. To achieve the best possible BLER performance, it is desirable to attain full diversity, i.e., $G_d = L$, where $L$ denotes the codeword length.

Phase rotation has been shown to be an effective approach for achieving full diversity when the transmitted signal is formed from a vector of modulation symbols \cite{Wang_26}. Motivated by this result, we apply phase rotation to the sparse vector $\bm{s}$ before the codebook spreading stage in the proposed SCMA-SVC scheme, as illustrated in Fig.~\ref{SCMA_SVC_Scheme}. 

In the following theorem, we show that a properly designed diagonal phase rotation $\bm{P}$ enables the proposed SCMA-SVC to achieve the full diversity order $L$ almost surely. 

\begin{thm}\label{thm:full_div_random_phase} Assume that all entries of $\bm{C}$ are nonzero, i.e., $c_{i,n} \neq 0$ for all $i,n$. Let $\bm{P}$ be the diagonal phase rotation matrix in (\ref{P_form}), whose phases $\{\phi_n\}_{n=0}^{N-1}$ are i.i.d. and continuously distributed over $[0,2\pi)$. Then the SCMA-SVC codebook $\mathcal{X}= \{\bm{x}=\bm{C}\bm{P}\bm{s} \mid \bm{s}\in\mathcal{S}\}$
achieves full diversity order $G_d=L$ almost surely.
\end{thm}


\begin{IEEEproof} Since $\bm{x}_l=\bm{C}\bm{P}\bm{s}_l$ and $\bm{x}_p=\bm{C}\bm{P}\bm{s}_p$, it follows that $\bm{x}_l-\bm{x}_p=\bm{C}\bm{P}(\bm{s}_l-\bm{s}_p)$. Define the difference set of the sparse vector set $\mathcal S$ as $\mathcal{E}=\{\bm{s}_l-\bm{s}_p \mid \bm{s}_l,\bm{s}_p\in\mathcal{S},\ l\neq p\}$. Consider two distinct codewords $\bm{x}_l$ and $\bm{x}_p$, and let 
\begin{equation}
\bm{e} = \bm{s}_l - \bm{s}_p \in \mathcal{E}.
\end{equation}
Define the $n$-th component-wise difference between the sparse vectors $\bm{s}_l$ and $\bm{s}_p$ as $\lambda_n=s_{l,n}-s_{p,n}$. Recall that the $i$-th component-wise difference between two codewords $\bm{x}_l$ and $\bm{x}_p$ is defined as $\Delta_i = x_{l,i}-x_{p,i}$. Under flat Rayleigh fading with independent fading across dimensions, the diversity order of the pairwise error probability is determined by the number of nonzero component-wise distances $\Delta_i$. Therefore, full diversity order $G_{d}=L$ is achieved if
\begin{equation}
\Delta_i \neq 0,
\quad
\forall i = 0,1,\dots,L-1,
\quad
\forall \bm{e} \in \mathcal{E}.
\label{eq:full_div_condition}
\end{equation}
Fix an arbitrary index $i \in\{0,1,\ldots,L-1\}$ and a nonzero difference vector
$\bm{e}\in \mathcal{E}$. Let its support set be $\mathcal{I}=\{n \mid \lambda_n \neq 0\}$ where $|\mathcal{I}|=t\ge 1$. From $\bm{x}_l-\bm{x}_p = \bm{C}\bm{P}(\bm{s}_l-\bm{s}_p)$ and $\lambda_n = s_{l,n}-s_{p,n}$, the $i$-th component-wise distance can be written as
\begin{equation}
\begin{aligned}
    \Delta_i=\sum_{n=0}^{N-1} c_{i,n}e^{j\phi_n}\lambda_n
=\sum_{n\in\mathcal{I}} c_{i,n}e^{j\phi_n}\lambda_n.
\end{aligned}
\end{equation}
Let $q\in\mathcal I$ be arbitrary. Then we can rewrite $\Delta_i$ as
\begin{equation}
\Delta_i=c_{i,q}\lambda_q e^{j\phi_q}+\sum_{n\in\mathcal{I}\setminus\{q\}}c_{i,n}\lambda_n e^{j\phi_n}.
\label{eq:isolate}
\end{equation}
If $\Delta_i = 0$, we obtain
\begin{equation}
e^{j\phi_q}=-\frac{\sum_{n\in\mathcal{I}\setminus\{q\}}c_{i,n}\lambda_n e^{j\phi_n}}{c_{i,q}\lambda_q}.
\label{eq:phi_constraint}
\end{equation}
Since $\lambda_q \neq 0$ and $c_{i,q} \neq 0$ by assumption, the denominator is nonzero. Hence the right-hand side of \eqref{eq:phi_constraint} is a well-defined complex constant, denoted by $u$. Therefore,
\begin{equation}
\phi_q = \arg(u) \quad (\mathrm{mod}\; 2\pi),
\label{eq:unique_phi}
\end{equation}
which corresponds to at most one value in $[0,2\pi)$. Thus, for each fixed pair $(i,\bm{e})$ and fixed configuration
of $\{\phi_n\}_{n\neq q}$, the event $\Delta_i = 0$ requires at most one specific value $\theta \in [0,2\pi)$ of $\phi_q$. Assume that
\begin{equation}
\phi_q \sim \mathcal{U}[0,2\pi),
\qquad
f_{\phi_q}(\varphi)=\frac{1}{2\pi}, \quad \varphi\in[0,2\pi),
\label{eq:uniform_phase}
\end{equation}
and that $\{\phi_n\}$ are mutually independent. Since $\phi_q$ is drawn from a continuous distribution over $[0,2\pi)$,
the probability that it takes any specific value is zero, i.e.,
\begin{equation}
\Pr(\phi_q = \theta) = 0.
\end{equation}
Consequently, $\Pr\!\left(\Delta_i = 0\;\middle|\;\{\phi_n\}_{n\neq q}\right)= 0$. Taking expectation over the remaining phase variables yields
\begin{equation}
\Pr\!\left(\Delta_i = 0\right)=\mathbb{E}_{\{\phi_n\}_{n\neq q}}\!\left[\Pr\!\left(\Delta_i = 0\;\middle|\;\{\phi_n\}_{n\neq q}\right)\right]=0.
\label{eq:unconditional_zero_refined}
\end{equation}
Since both $i$ and $\bm{e}$ range over finite sets, the number of events $\{\Delta_i = 0\}$ is finite. By the union bound,
\begin{equation}
\Pr\!\left(\bigcup_{i,\bm{e}}\{ \Delta_i = 0 \}\right)\le\sum_{i,\bm{e}}
\Pr\!\left(\Delta_i = 0\right)=0.
\label{eq:union_zero_refined}
\end{equation}
Therefore, condition \eqref{eq:full_div_condition} holds with probability $1$. Hence, the proposed SCMA-SVC achieves full diversity order $L$ almost surely.
\end{IEEEproof}

\section{Proposed Low-Complexity Decoding for SCMA-SVC}\label{Sec:SCMA_SVC_decoder}
In this section, we develop a low-complexity near-ML decoder for SCMA-SVC that leverages the structured sparsity of the SCMA-SVC codebook to reduce the complexity of ML detection. A generalized likelihood ratio test (GLRT)-based metric is employed to efficiently identify the most likely sparse pattern. Subsequently, the search for the estimated codeword $\hat{\bm{x}}$ is restricted to a highly correlated subset of candidate codewords, rather than exhaustively searching over all possible codewords as in conventional ML detection.

Generally, the ML decoder finds the codeword $\hat{\bm x}$ in the transmitter codebook $\mathcal{X}=\{\bm{x}_l \mid l=0,1,\ldots,{2^{b_t}}-1\}$ whose Euclidean distance to the received signal $\bm y$ is minimized, i.e.,
\begin{equation}\label{ML_Decoder}
\begin{aligned} 
    \hat{\bm{x}} &= \arg \min\limits_{l=0,1,\ldots,2^{b_t}-1}\| \bm{x}_{l}-\bm{y}\|_{2}^{2}.
\end{aligned}
\end{equation}
The search complexity of ML detection grows exponentially with the number of information bits $b_t$, requiring exhaustive evaluation
over all $2^{b_t}$ codewords. Even for moderate values of~$b_t$, such exhaustive search becomes undesirable under stringent latency constraints. Therefore, our objective is to develop an efficient search strategy that significantly reduces the decoding complexity while maintaining near-ML performance. Specifically, instead of exhaustively searching over all $2^{b_t}$ codewords, the detection is performed only within a reduced subset of highly reliable codeword candidates.

The received signal is modeled as
\begin{equation}
\bm{y} = \bm{H}\bm{x} + \bm{n}= \bm{H}\bm{C}\bm{P}\bm{s} + \bm{n}\triangleq \bm{\Phi}\bm{s} + \bm{n},
\end{equation}
where $\bm{\Phi}\triangleq\bm{H}\bm{C}\bm{P}$ denotes the effective channel-aware spreading dictionary. Due to the structured sparsity of $\bm{s}$, only $K$ columns of $\bm{\Phi}$ are active for each codeword, where $\bm{\Phi}=\begin{bmatrix}\bm{\phi}_{0},\bm{\phi}_{1},\ldots,\bm{\phi}_{N-1} \end{bmatrix}$. Let $\mathcal{J}=\{u_{0},u_{1},\ldots,u_{K-1}\}$ denote a candidate set of $K$ active indices of the sparse vector $\bm{s}$. We denote by $\bm{\Phi}_{\mathcal{J}}$ the submatrix formed by selecting the columns of $\bm{\Phi}$ indexed by $\mathcal{J}$. For example, if $\bm{s}= (0,0,1,1,0,0,0,0,1,0,0,0)^T$, then $\mathcal{J}=\{2,3,8\}$ and $\bm{\Phi}_{\mathcal{J}}=[\bm{\phi}_{2},\bm{\phi}_{3},\bm{\phi}_{8}]$. 

Under hypothesis $\mathcal{J}$, the received signal is confined to the $K$-dimensional subspace spanned by the columns of $\bm{\Phi}_{\mathcal{J}}$ and can be written as
\begin{equation}
\bm{y} = \bm{\Phi}_{\mathcal{J}} \bm{s}_{\mathcal{J}} + \bm{n},
\qquad 
\bm{n} \sim \mathcal{CN}(\bm{0},\sigma^{2}\bm{I}),
\end{equation}
where $\bm{s}_{\mathcal{J}}\in\mathbb{C}^{K\times 1}$ contains the unknown symbol coefficients. For each hypothesis ${\mathcal{J}}$, the vector $\bm{s}_{\mathcal{J}}$ is treated as a nuisance parameter. The GLRT metric is defined as
\begin{equation}\label{eq:glrt_def}
\Lambda({\mathcal{J}})
=
\max_{\bm{s}_{\mathcal{J}}}
\; p(\bm{y} \mid \bm{s}_{\mathcal{J}}, \bm{\Phi}_{\mathcal{J}}),
\end{equation}
where the conditional likelihood is
\begin{equation}\label{eq:likelihood}
p(\bm{y} \mid \bm{s}_{\mathcal{J}}, \bm{\Phi}_{\mathcal{J}})
=
\frac{1}{(\pi \sigma^2)^L}
\exp\!\left(
-\frac{1}{\sigma^2}
\left\| \bm{y} - \bm{\Phi}_{\mathcal{J}} \bm{s}_{\mathcal{J}} \right\|_{2}^{2}
\right).
\end{equation}
Since the exponential function in \eqref{eq:likelihood} is monotonic,
maximizing the likelihood in \eqref{eq:glrt_def} is equivalent to solving
\begin{equation}
\min_{\bm{s}_{\mathcal{J}}}
\|\bm{y}-\bm{\Phi}_{\mathcal{J}}\bm{s}_{\mathcal{J}}\|_2^2.
\label{eq:ls_problem}
\end{equation}
The corresponding least-squares estimate is 
\begin{equation}
\hat{\bm{s}}_{\mathcal{J}}
=
(\bm{\Phi}_{\mathcal{J}}^{\mathrm H}\bm{\Phi}_{\mathcal{J}})^{-1}
\bm{\Phi}_{\mathcal{J}}^{\mathrm H}\bm{y}.
\label{eq:ls_solution}
\end{equation}
Using \eqref{eq:ls_solution}, the minimum of
\eqref{eq:ls_problem} is attained at
$\bm{s}_{\mathcal{J}} = \hat{\bm{s}}_{\mathcal{J}}$, yielding
\begin{equation}
\min_{\bm{s}_{\mathcal{J}}}
\|\bm{y}-\bm{\Phi}_{\mathcal{J}}\bm{s}_{\mathcal{J}}\|_2^2
=
\|\bm{y}-\bm{\Phi}_{\mathcal{J}}\hat{\bm{s}}_{\mathcal{J}}\|_2^2.
\label{eq:residual_step1}
\end{equation}
Substituting \eqref{eq:ls_solution} into the signal model yields
\begin{equation}
\bm{\Phi}_{\mathcal{J}} \hat{\bm{s}}_{\mathcal{J}}
=\bm{\Phi}_{\mathcal{J}}
(\bm{\Phi}_{\mathcal{J}}^{\mathrm H}\bm{\Phi}_{\mathcal{J}})^{-1}
\bm{\Phi}_{\mathcal{J}}^{\mathrm H} \bm{y}
\triangleq \bm{Z}_{\mathcal{J}} \bm{y},
\label{eq:projection_matrix}
\end{equation}
where $\bm{Z}_{\mathcal{J}}$ denotes the orthogonal projection matrix onto
$\mathrm{span}(\bm{\Phi}_{\mathcal{J}})$. Substituting \eqref{eq:projection_matrix} into
\eqref{eq:residual_step1}, we obtain
\begin{equation}
\|\bm{y}-\bm{\Phi}_{\mathcal{J}}\hat{\bm{s}}_{\mathcal{J}}\|_2^2
=
\|\bm{y}-\bm{Z}_{\mathcal{J}}\bm{y}\|_2^2
=
\|\bm{y}\|_2^2
-
\bm{y}^{\mathrm H}\bm{Z}_{\mathcal{J}}\bm{y}.
\label{eq:residual_energy}
\end{equation}
Since $\|\bm y\|_2^2$ is independent of $\mathcal{J}$, minimizing \eqref{eq:ls_problem} is therefore equivalent to maximizing $\bm y^{\mathrm H}\bm Z_{\mathcal{J}}\bm y$. Therefore, the GLRT metric of \eqref{eq:glrt_def} can be written as
\begin{equation}
\Lambda({\mathcal{J}})=\bm{y}^{\mathrm H}\bm{Z}_{\mathcal{J}}\bm{y}
=
\bm{y}^{\mathrm H}
\bm{\Phi}_{\mathcal{J}}
(\bm{\Phi}_{\mathcal{J}}^{\mathrm H}\bm{\Phi}_{\mathcal{J}})^{-1}
\bm{\Phi}_{\mathcal{J}}^{\mathrm H}
\bm{y}.
\label{eq:glrt_explicit_o}
\end{equation}

\subsection{Proposed Two-Stage Near-ML Decoder for SCMA-SVC}


\begin{algorithm}[t]
\caption{Two-Stage Near-ML Decoding for SCMA-SVC}
\label{TS_SCMA_algorithm}
\begin{algorithmic}[1]
\REQUIRE Received vector $\bm{y}\in\mathbb{C}^{L\times 1}$, channel matrix $\bm{H}$, channel-aware spreading dictionary $\bm{\Phi}$, candidate index sets $\mathcal{G}$, and selection parameter $\beta$.
\ENSURE Estimated information bits $\hat{\bm{m}}$.

\vspace{0.3em}

\vspace{0.3em}
\STATE \textbf{Stage 1 (GLRT-Based Index Selection):}
\FOR{each group ${\mathcal{J}}\in\mathcal{G}$ with index set $\{u_0,u_1,\ldots,u_{K-1}\}$}
    \STATE Form $\bm{\Phi}_{\mathcal{J}} \triangleq [\bm{\phi}_{u_0},\bm{\phi}_{u_1},\ldots,\bm{\phi}_{u_{K-1}}]\in\mathbb{C}^{L\times K}$.
    \STATE Compute $\bm{v}_{\mathcal{J}} \triangleq \bm{\Phi}_{\mathcal{J}}^H\bm{y}$ and $\bm{R}_{\mathcal{J}} \triangleq \bm{\Phi}_{\mathcal{J}}^H\bm{\Phi}_{\mathcal{J}}$.
    \STATE Evaluate the GLRT score
    \begin{equation*}
    \begin{aligned}
    &\Lambda({\mathcal{J}}) = \bm{v}_{\mathcal{J}}^H \bm{R}_{\mathcal{J}}^{-1}\bm{v}_{\mathcal{J}}\\
&(\text{equivalently } \Lambda(\mathcal{J})=\bm{y}^{\mathrm H}
\bm{\Phi}_{\mathcal{J}}
(\bm{\Phi}_{\mathcal{J}}^{\mathrm H}\bm{\Phi}_{\mathcal{J}})^{-1}
\bm{\Phi}_{\mathcal{J}}^{\mathrm H}
\bm{y}).
    \end{aligned}
    \end{equation*}
\ENDFOR
\STATE $\mathcal K_\beta=\{j_0,j_1,\ldots,j_{\beta-1}\} \leftarrow$ the $\beta$ indices $j$ with the largest $\Lambda(\mathcal J)$.
\vspace{0.3em}
\STATE \textbf{Stage 2 (Fine Near-ML Search within Selected Groups):}
\STATE Initialize $d_{\min}\leftarrow+\infty$ and $\hat{m}\leftarrow 0$.
\FOR{$j \in \mathcal{K}_\beta$}
    \FOR{$q=0,\ldots,Q-1$}
        \STATE $m \leftarrow jQ+q$
        \STATE Compute the Euclidean distance
        \begin{equation*}
            d(m)=\|\bm{y}-\bm{H}{\bm{x}}_m\|_{2}^{2}.
        \end{equation*}
        \IF{$d(m)<d_{\min}$}
            \STATE $d_{\min}\leftarrow d(m)$, $\hat{m}\leftarrow m$.
        \ENDIF
    \ENDFOR
\ENDFOR
\STATE \textbf{return} $\hat{m}$.
\STATE Map $\hat{m}$ to the estimated information bits $\hat{\bm{m}}$.
\end{algorithmic}
\end{algorithm}

Based on the GLRT metric in \eqref{eq:glrt_explicit_o}, we propose a two-stage near-ML decoder for SCMA-SVC. The overall decoding process includes two stages as follows:

\subsubsection{Stage 1 (GLRT-Based Index Selection)}

Let $\mathcal{G}=\{{\mathcal{J}}_0,{\mathcal{J}}_1,\ldots,{\mathcal{J}}_{J-1}\}$ denote the set of all candidate index sets known at the receiver in SCMA-SVC. For each index set ${\mathcal{J}}\in\mathcal{G}$, we evaluate the GLRT metric
\begin{equation}
\Lambda({\mathcal{J}})
=
\bm{y}^{\mathrm H}
\bm{\Phi}_{\mathcal{J}}
(\bm{\Phi}_{\mathcal{J}}^{\mathrm H}\bm{\Phi}_{\mathcal{J}})^{-1}
\bm{\Phi}_{\mathcal{J}}^{\mathrm H}
\bm{y}.
\label{eq:stage1_metric}
\end{equation}
Then, the $\beta$ indices ${\mathcal{J}}\in\mathcal{G}$ with the largest GLRT scores $\Lambda({\mathcal{J}})$ are selected and denoted by $\mathcal{K}_\beta$, where $\beta \ll J$. This corresponds to selecting the index sets whose associated subspaces capture the largest projection energy of the received signal.

\subsubsection{Stage 2 (Fine Near-ML Search within Selected Groups)}

After {\it Stage~1}, a restricted ML detection is performed over the codewords associated with the selected index sets $\mathcal{K}_\beta$. Since each index set corresponds to $Q$ constellation symbols, the resulting search space contains $Q\beta$ candidate codewords. Compared to the conventional ML detector in \eqref{ML_Decoder}, the search space is significantly reduced from $2^{b_t}$ to $Q\beta$.

The detailed implementation steps of the proposed two-stage decoder are summarized in {\it Algorithm~\ref{TS_SCMA_algorithm}}.

\subsection{Complexity Analysis of Two-Stage Near-ML Decoder for SCMA-SVC}

We analyze the computational complexity of the proposed two-stage near-ML decoder and compare it with the conventional ML detector. For each index set $\mathcal{J}$, since the sparsity level $K$ is fixed and small in SCMA-SVC, the per-index-set complexity scales linearly with $L$, i.e., $O(L)$. Since the GLRT metric is evaluated for all $J$ candidate index sets, the total complexity of Stage 1 is $C_{\text{Stage 1}} = O(JL)$. After Stage~1, only the top-$\beta$ candidate indices are retained, where $\beta \ll J$. Each selected index corresponds to $Q$ constellation symbols, resulting in $Q \beta$ candidate codewords. For each candidate codeword, an $L$-dimensional Euclidean distance is evaluated. Hence,
$C_{\text{Stage 2}} = O(Q\beta L)$. Therefore, the overall complexity is $C_{\text{total}}=C_{\text{Stage 1}}+C_{\text{Stage 2}} = O(JL + Q\beta L)$. In contrast, the conventional ML detector evaluates all $JQ = 2^{b_t}$ candidate codewords, each requiring an $L$-dimensional distance computation, resulting in a complexity of $O(QJ L)$. When $\beta \ll J$, the proposed two-stage decoder reduces the search complexity from $O(QJ L)$ to $O(Q\beta L+JL)$ while preserving near-ML detection performance.

\section{Simulation Results}\label{Sec:Simulation}

\begin{figure*}[!t]
    \includegraphics[width=0.9\textwidth]{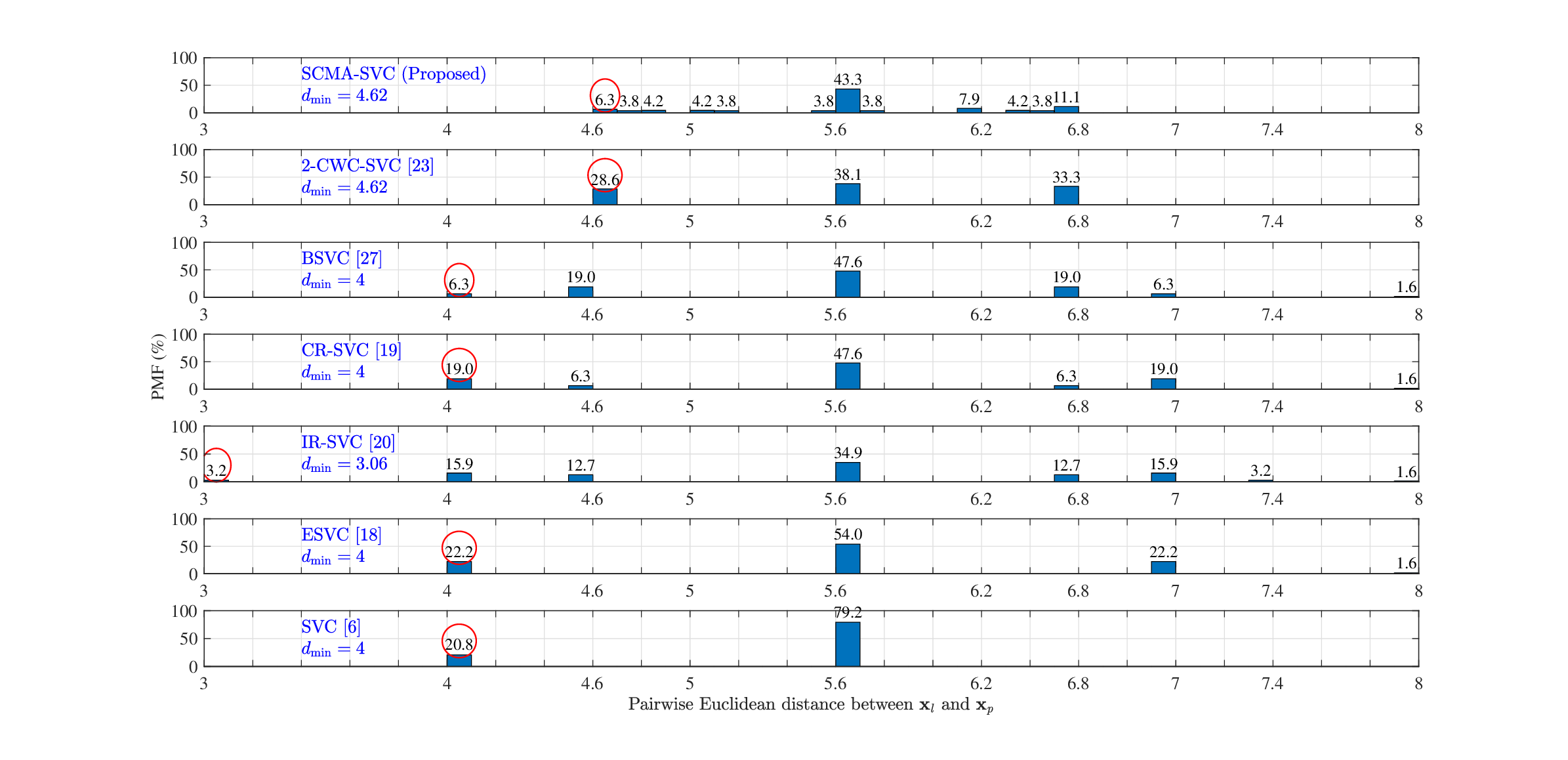}
    \caption{Euclidean distance spectrum of the codewords for the proposed SCMA-SVC and several SVC variants, illustrated by the PMFs of the 
pairwise Euclidean distances $\|\bm{x}_l-\bm{x}_p\|_2$ between distinct codewords. The MED of each codebook is indicated by~$d_{\text{min}}$, while the red-circled percentages denote the proportion of codeword pairs whose pairwise Euclidean distance equals the MED. ($b_t=6$, $L=16$)} \label{Distance_Comparison_Sim_1}
\end{figure*}
In this section, we evaluate the BLER performance of the proposed SCMA-SVC scheme over Gaussian and Rayleigh fading channels under ML decoding to characterize the inherent performance of the codebook. We compare the proposed SCMA-SVC with several existing SVC variants, including SVC \cite{Ji_18_first}, ESVC \cite{Kim_20}, CR-SVC \cite{Zhang_22}, IR-SVC \cite{Zhang_23}, BSVC \cite{Zhang_25_block}, and 2-CWC-SVC \cite{Mow_24_CWC}. Normal approximation\text{\,\begingroup\renewcommand{\thefootnote}{\fnsymbol{footnote}}\footnotemark[3]\endgroup} \cite{Vincent_10} is also included as a finite-blocklength performance benchmark. In the simulations, the channel matrix $\bm H$ is assumed to be known at the decoder for all schemes. We first consider a short-packet configuration with $b_t=6$ information bits and codeword length of $L=16$. We employ the codebook $\mathcal{X}$ of the proposed SCMA-SVC from {\it Example~\ref{eg_SCMA-SVC_6_bits}}. The only modification is that a diagonal phase rotation matrix $\bm P$ is introduced, where $\phi_0=6.04, \phi_1=3.39, \phi_2=0.19, \phi_3=4.37, \phi_4=3.26, \phi_5=0.37, \phi_6=5.59, \phi_7=2.07, \phi_8=1.44, \phi_9=0.71, \phi_{10}=1.95$, and $\phi_{11}=1.43$ according to {\it Theorem~\ref{thm:full_div_random_phase}}. For fair comparison, although different SVC variants employ different sparse mappings, which lead to different codebook dimensions, we maintain a consistent codeword length across all methods by adopting the same orthogonal codebook matrix $\bm{C}$. The corresponding codebook matrix is obtained by taking the first $N$ columns of $\bm{C}$. The proposed SCMA-SVC employs a multi-dimensional constellation of size $Q=4$, while the SVC, ESVC, CR-SVC, IR-SVC, and BSVC schemes adopt $4$-QAM constellations. For the 2-CWC-SVC scheme, a 2-CWC can be constructed from a CWC by assigning $\{\pm1\}$-valued symbols to each nonzero entry of the codeword~\cite{Mow_24_CWC}. In our setting, since the sparse pattern vectors in the indicator matrix $\bm{F}$ of {\it Example~\ref{eg_SCMA-SVC_6_bits}} can be regarded as a binary CWC, the sparse vector set $\mathcal{S}$ is constructed by assigning the columns of the constellation matrix $\bm{A}$ to the nonzero positions specified by $\bm{F}$ where
\begin{equation}
\bm{A}=
\begin{bmatrix}
1 & 1 & -1 & -1\\
1 & -1 & -1 & 1\\
1 & -1 & 1 & -1
\end{bmatrix}.
\end{equation}
\begin{figure}[!t]
    \centering
    \begin{center}
        \includegraphics[width=0.45\textwidth]{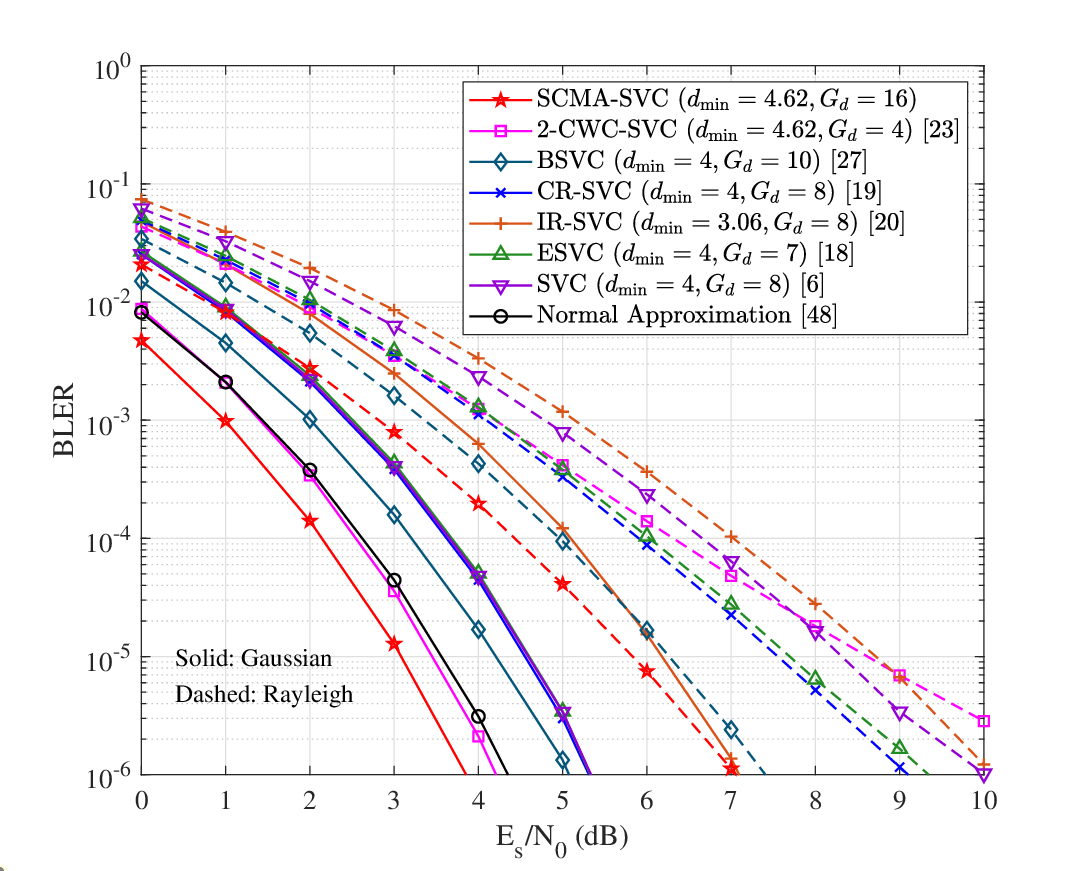}
        \caption{BLER performance comparison of different SVC-based codebooks under ML decoding over Gaussian and Rayleigh channels with $b_t=6$ information bits and codeword length $L=16$.}\label{BLER_SCMA_SVC_6_16}
    \end{center}
\end{figure}
\begingroup
\renewcommand{\thefootnote}{\fnsymbol{footnote}}
\footnotetext[3]{The normal approximation is a finite-blocklength information-theoretic approximation that estimates the maximum achievable coding rate for a given blocklength and target error probability~\cite{Vincent_10}.}
\endgroup 
\begin{figure*}[!t]
    \includegraphics[width=0.9\textwidth]{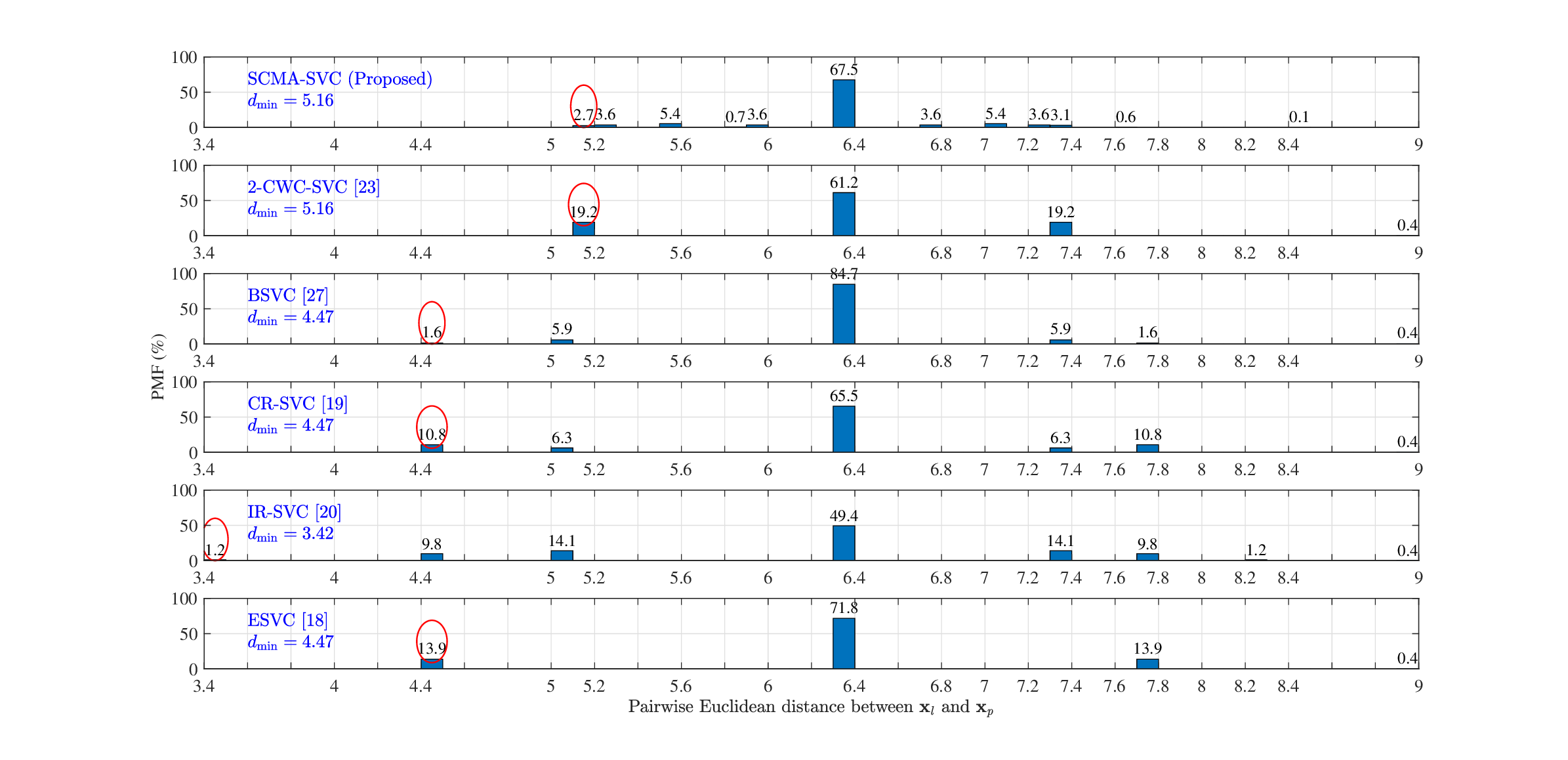}
    \caption{Euclidean distance spectrum of the codewords for the proposed SCMA-SVC and several SVC variants, illustrated by the PMFs of the 
pairwise Euclidean distances $\|\bm{x}_l-\bm{x}_p\|_2$ between distinct codewords. The MED of each codebook is indicated by~$d_{\text{min}}$, while the red-circled percentages denote the proportion of codeword pairs whose pairwise Euclidean distance equals the MED. ($b_t=8$, $L=20$)} \label{Distance_Comparison_Sim_2}
\end{figure*}
\hspace{-1em}The resulting sparse vector set $\mathcal{S}$ corresponds to a structured subset of the 2-CWC construction in~\cite{Mow_24_CWC}. Fig.~\ref{Distance_Comparison_Sim_1} shows the distribution of the pairwise Euclidean distances $\|\bm{x}_l-\bm{x}_p\|_2$ between distinct codewords for the proposed SCMA-SVC and several SVC variants, which characterizes the Euclidean distance spectrum of the corresponding codebooks. The distribution is represented by the probability mass function (PMF). By enumerating all distinct codeword pairs in the codebook~$\mathcal{X}$, the Euclidean distance is computed for each pair. The PMF is then obtained by counting the number of pairs associated with each distance value and normalizing the counts by the total number of pairs. Both the proposed SCMA-SVC and the 2-CWC-SVC achieve the largest MED $d_{\min}=4.62$ among all considered SVC schemes. This is because the codebooks in these two schemes are constructed with consideration of the MED, whereas the other SVC variants do not consider the MED in their codebook design. Although the 2-CWC-SVC attains the same MED, the number of codeword pairs corresponding to $d_{\min}$ in SCMA-SVC is smaller, indicating a more favorable Euclidean distance spectrum.

Fig.~\ref{BLER_SCMA_SVC_6_16} demonstrates the BLER performance of the proposed SCMA-SVC and several SVC variants under the ML decoder over Gaussian and Rayleigh fading channels. We observe that the  proposed SCMA-SVC outperforms the existing SVC variants across the entire SNR range in both Gaussian and Rayleigh fading channels. For the Gaussian channel, the proposed SCMA-SVC achieves the largest MED and a more favorable distance spectrum, thereby lowering the pairwise error probability. For a target BLER of $10^{-5}$, the proposed SCMA-SVC achieves $0.37$ dB and $1.12$ dB SNR gains over 2-CWC-SVC and BSVC, respectively. For the Rayleigh fading channel, the proposed SCMA-SVC also achieves superior BLER performance compared with the existing SVC variants. This is because the proposed SCMA-SVC achieves the full diversity order $G_d = 16$ through phase rotation. In contrast, the other SVC schemes do not guarantee full diversity in fading channels.

\begin{figure}[!t]
    \centering
    \begin{center}
        \includegraphics[width=0.45\textwidth]{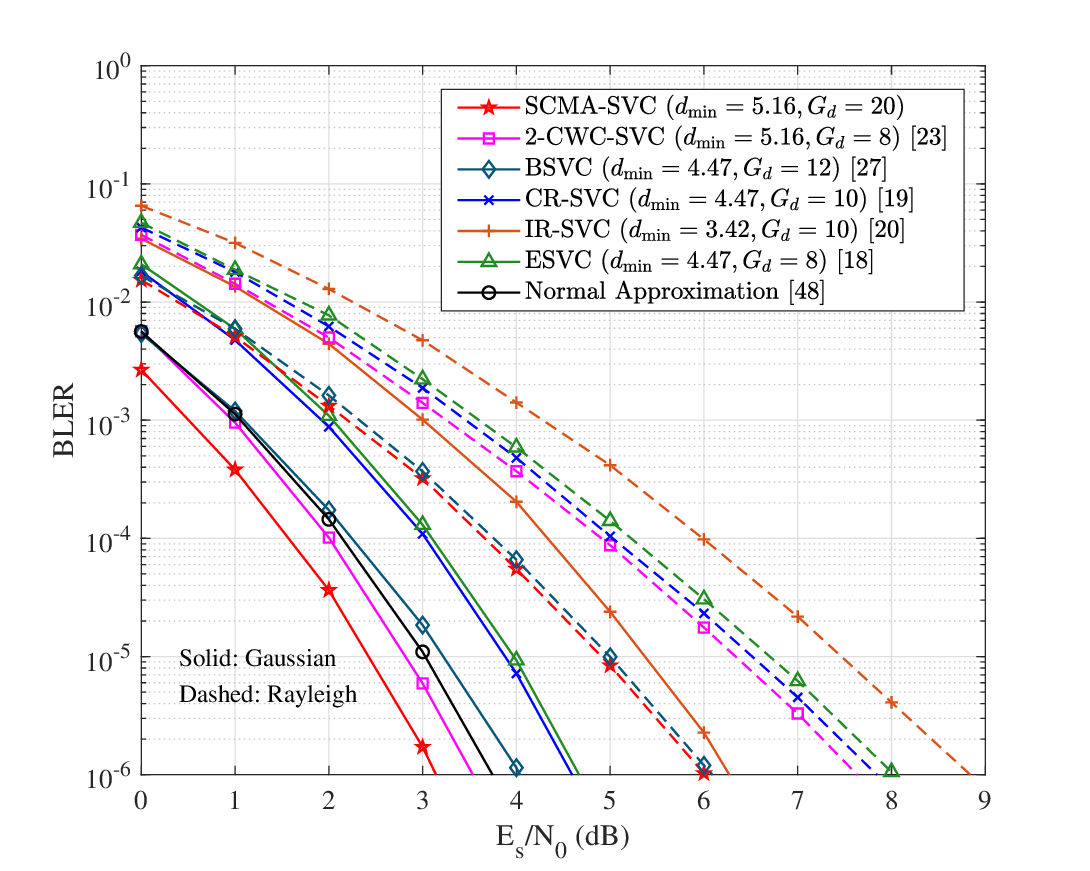}
        \caption{BLER performance comparison of different SVC-based codebooks under ML decoding over Gaussian and Rayleigh channels with $b_t=8$ information bits and codeword length $L=20$.}\label{BLER_SCMA_SVC_8_20}
    \end{center}
\end{figure}
Next, we consider another configuration with $b_t=8$ information bits and codeword length of $L=20$. The proposed SCMA-SVC scheme is simulated with parameters $b_0=5$, $b_1=3$, $N=20$, $J=32$, $K=3$, $Q=8$, and $\alpha=20$. The indicator matrix $\bm{F}$ of size $20 \times 32$ with girth $6$ is adopted. The multi-dimensional constellation $\bm{A}_{\text{MC}}$ of size $3 \times 8$, with an MED $d_{\min}(\mathcal{A}_{\text{MC}})=1.313$ as given in (\ref{A_MC_3_8}), is employed. A bipolar orthogonal codebook matrix $\bm{C}$ of size $20 \times 20$ satisfying $\bm{C}^{H}\bm{C}=20\bm{I}_{20}$ and a phase rotation matrix $\bm P$ of size $20 \times 20$ are employed. As a result, the generated codebook $\mathcal{X}$ achieves an MED $d_{\min}=\min\{1.155,1.313\} \times \sqrt{20}=5.16$ (refer to {\it Remark~\ref{Remark_SCMA_dim_form}}). We compare the performance of SCMA-SVC with ESVC \cite{Kim_20}, CR-SVC \cite{Zhang_22}, IR-SVC \cite{Zhang_23}, BSVC \cite{Zhang_25_block}, and 2-CWC-SVC \cite{Mow_24_CWC}, constructed in the same manner as in the previous simulation. As shown in Fig.~\ref{Distance_Comparison_Sim_2}, the proposed SCMA-SVC achieves the largest MED and exhibits a more favorable Euclidean distance spectrum among all considered SVC schemes, with fewer codeword pairs occurring at the minimum distance. Therefore, over the Gaussian channel, SCMA-SVC achieves superior BLER performance compared with the existing SVC variants, as shown in Fig.~\ref{BLER_SCMA_SVC_8_20}. For the Rayleigh fading channel, the proposed SCMA-SVC also achieves the best BLER performance among all considered SVC schemes. This is because SCMA-SVC achieves the full diversity order $G_d=20$ through phase rotation. Therefore, SCMA-SVC provides more reliable performance under both Gaussian and Rayleigh fading channels.
\begin{figure*}[!t]
\centering
\begin{equation}\label{A_MC_3_8}
\resizebox{\textwidth}{!}{$
\bm{A}_{\text{MC}} =
\begin{bmatrix}
0.4082 - 0.4082i &  0.5774 - 0.0000i & 0.4082 + 0.4082i & 0.0000 + 0.5774i & -0.4082 + 0.4082i & -0.5774 + 0.0000i & -0.4082 - 0.4082i & -0.0000 - 0.5774i \\
0.2209 - 0.5334i & -0.2209 + 0.5334i & -0.5334 + 0.2209i &  0.5334 - 0.2209i & 0.2209 + 0.5334i & -0.2209 - 0.5334i & -0.5334 - 0.2209i &  0.5334 + 0.2209i\\
0.5774 - 0.0000i &  0.4082 - 0.4082i & -0.4082 + 0.4082i & -0.5774 + 0.0000i & 0.4082 + 0.4082i  & 0.0000 + 0.5774i & -0.0000 - 0.5774i & -0.4082 - 0.4082i
\end{bmatrix}
$}.
\end{equation}
\end{figure*}

Finally, we evaluate the BLER performance of the proposed two-stage near-ML decoder described in {\it Algorithm~\ref{TS_SCMA_algorithm}}. The optimal ML detector is included as a benchmark. The same configuration with $b_t=6$ information bits and codeword length $L=16$ is adopted. Fig.~\ref{TS_BLER} shows the BLER performance of the proposed two-stage near-ML decoder for different values of $\beta$ under both Gaussian and Rayleigh fading channels. It can be observed that the performance of the proposed decoder approaches that of the optimal ML detector as $\beta$ increases. This is because a larger $\beta$ retains more candidate index sets after Stage 1, which increases the probability that the correct index set is included in the candidate set. Moreover, the proposed low-complexity two-stage near-ML decoder incurs only a slight SNR loss compared with the optimal ML detector. At a target BLER of $10^{-5}$, the SNR loss is $0.07$ dB for the Gaussian channel with $\beta=3$, and $0.18$ dB for the Rayleigh fading channel with $\beta=5$. These results demonstrate that the proposed two-stage decoder can effectively reduce the search complexity while maintaining near-ML detection performance by properly selecting the parameter $\beta$ and exploiting the structured SCMA-SVC codebook design.

\section{Conclusion and Future Work}\label{Sec:Conclusion}
\begin{figure}[!t]
    \centering
    \hspace{-0.6em}
    \includegraphics[width=0.45\textwidth]{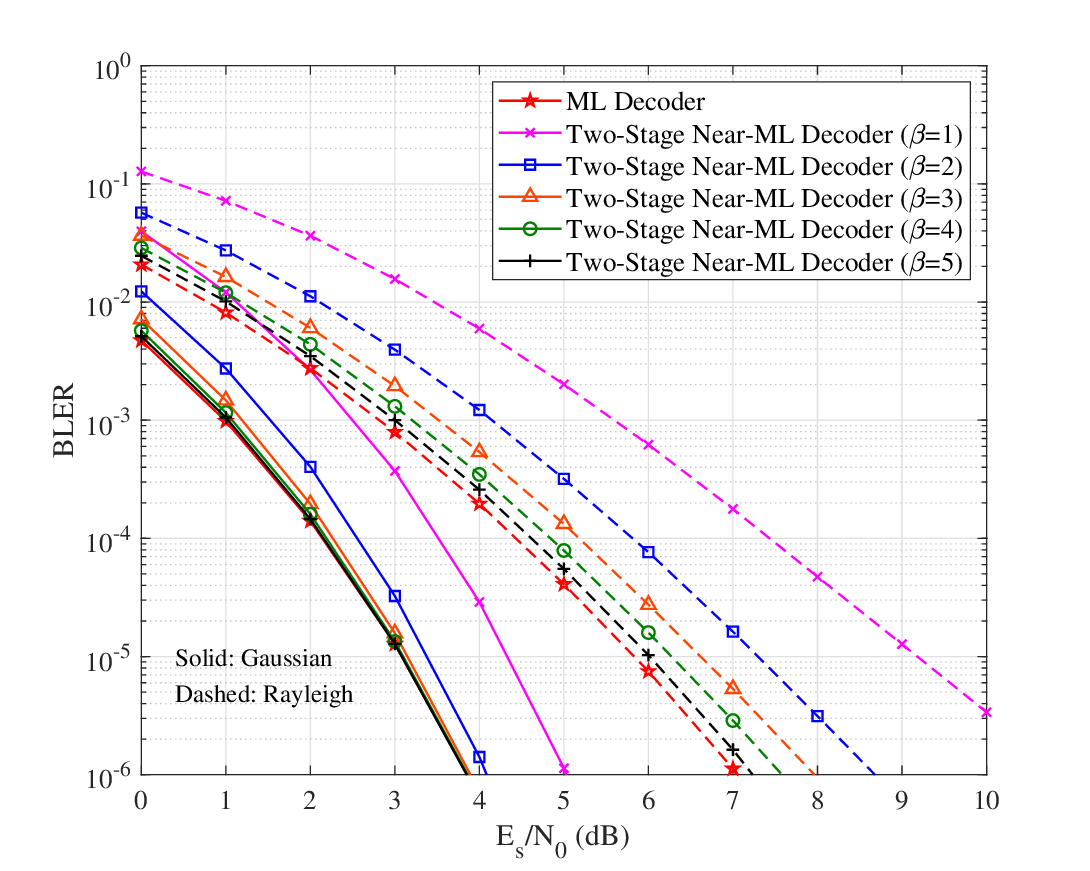}
    \caption{BLER performance comparison of the proposed two-stage near-ML decoder for SCMA-SVC with different values of $\beta$.}\label{TS_BLER}
\end{figure}
This work is dedicated to the study of SCMA inspired SVC for enhanced short-packet URLLC transmission. The major novelty of this work stems from the extraction of multiuser coding gain and constellation shaping gain of SCMA which has not been reported in the literature before, to the best of our knowledge. First of all, the relationship between the MED of the sparse vector set and that of the transmitted codebook is established in {\it Theorem~\ref{dmin_theorem}}. Such a relationship has provided a useful design insight that, when an orthogonal codebook matrix is employed, both the MED and the Euclidean distance spectrum of the transmitted codebook are determined by the sparse vector design.  By employing a large-girth indicator matrix and a distance-optimized multi-dimensional constellation, the MED of the resulting SCMA-SVC codebook can be precisely characterized by the sparse pattern structure and the multi-dimensional mother constellation as mentioned in {\it Remark~\ref{Remark_SCMA_dim_form}}. Furthermore, by introducing random phase rotation, it has been shown that the proposed SCMA-SVC scheme achieves full diversity order over Rayleigh fading channels, as shown in {\it Theorem~\ref{thm:full_div_random_phase}}. As a result, the constructed codebook simultaneously possesses a large MED and full diversity, which significantly improves reliability over both Gaussian and Rayleigh fading channels. To address the high computational complexity of exhaustive ML detection, we have developed a low-complexity near-ML decoder by applying the GLRT with the structured sparsity of the SCMA-SVC codebook. The proposed two-stage decoder significantly reduces the number of Euclidean distance evaluations while maintaining near-ML detection performance. Simulation results have demonstrated that the proposed SCMA-SVC scheme achieves superior BLER performance compared with several existing SVC variants over both Gaussian and Rayleigh fading channels. These results confirm that the proposed SCMA-SVC framework provides an effective coding approach for enhancing the reliability of short-packet transmissions in URLLC scenarios. 

{\it Future Work:} Although the indicator matrices of SCMA-SVC are designed with a large girth to enhance sparse pattern separation and improve the Euclidean distance spectrum, the structural design of indicator matrices remains largely unexplored. An interesting future research direction is to investigate structured indicator matrix designs and study their impact on the decoding performance of SCMA-SVC.

%




\balance
\bibliographystyle{IEEEtran}
\bibliography{SCMA_SVC}

\vfill

\end{document}